\documentclass[11pt]{article}
\input{_style.sty}

\title{
PureSuperQMA(exp) = BellPureSymQMA(poly) = QMA \\ 
via Dimension-Free Bosonic Argmax}
\date{\today}

\author{William Gay\thanks{{\tt University of Illinois, Urbana-Champaign}. {\tt whgay2@illinois.edu}. }  
\and Fernando Granha Jeronimo\thanks{{\tt University of Illinois, Urbana-Champaign}. {\tt granha@illinois.edu}. }
\and Lenny Liu\thanks{{\tt University of Illinois, Urbana-Champaign}. {\tt hengyu2@illinois.edu}. }
\and Itai Leigh\thanks{{\tt Tel-Aviv University}. {\tt itai.leigh@mail.huji.ac.il}. }
\and Pei Wu\thanks{{\tt Penn State University}. {\tt pei.wu@psu.edu}. }
\and Haochen Xu\thanks{{\tt Penn State University}. {\tt hpx5065@psu.edu}. }
}

\begin{document}

\maketitle
\begin{abstract}
Pure-state consistency problems naturally lead to quantum proof systems in which a single pure witness must satisfy many acceptance constraints. The corresponding class $\mathsf{PureSuperQMA}$ was previously known to lie between $\mathsf{QMA}$ and $\mathsf{QMA}(2)$, and Kamminga and Rudolph (ITCS'26) conjectured that both containments are strict. In this paper, we prove the following surprising complexity collapses 
\[
    \mathsf{QMA} = \mathsf{PureSuperQMA} = \mathsf{PureSuperQMA}(\mathrm{exp}) = \mathsf{BellPureSymQMA}(\poly)
\]
Here $\mathsf{PureSuperQMA}(\mathrm{exp})$ allows exponentially many checks which are uniformly indexed and efficiently generated, while requiring an inverse-polynomial violation margin and an inverse-polynomial fraction of violated checks for the NO cases.
$\mathsf{BellPureSymQMA}(\poly)$ is a related model that requires the prover to give the verifier polynomially many copies of a pure state, which the verifier measures separately with logarithmic output length for each local measurement, before processing the outcomes jointly.

The main technical ingredient is a dimension-free stability bound for symmetric tensor states. Our simulations use polynomially many witness registers and combine a
random-pair SWAP test with a permutation-invariant lift of the original
verification procedure. The key step is to show that, on the symmetric subspace, the extremal verification value is close to that of some tensor-power witness with dimension-independent error. We prove this by choosing the tensor power of maximum overlap with an extremal symmetric eigenvector, so the resulting optimality eliminates the first-order error, while the stability bound controls all higher-order errors. Applying this argument to the two verification models yields both simulations, as a consequence, exact $k$-local pure-state consistency is $\mathsf{QMA}$-complete for every fixed $k\ge2$, and so are the corresponding exact bosonic and fermionic pure $N$-representability problems.

\end{abstract}
\newpage

\tableofcontents
\newpage

\section{Introduction}

Much of the experimentally accessible information about a many-body quantum
system is local. A molecule may contain many interacting particles, yet an
experiment typically probes only a few of them at a time. Likewise, the terms
of a physical Hamiltonian usually act on small subsets of the particles. The energy of such a local Hamiltonian $H=\sum_i H_i$ is determined by the reduced
density matrices on the supports of the $H_i$'s. This suggests optimizing over
these much smaller matrices instead of the exponentially large global state.
The density matrices cannot be varied independently, however, leading to the following local-to-global problem: 

{\centering Given density matrices $\rho_1,\ldots,\rho_m$ on small, possibly
overlapping sets of qubits, do they arise as the reduced states of one global
state?\par} 

This is the \emph{Consistency of Local Density Matrices} problem
(\textsf{CLDM}), also often called the quantum marginal problem. For systems of
identical particles, the corresponding compatibility question is known as
$N$-representability. These questions
have played a central role in quantum chemistry since the early work on
fermionic density matrices
\cite{Coleman1963,Klyachko2006}.

Liu placed mixed-state \textsf{CLDM} in $\mathsf{QMA}$, the quantum analogue of
$\mathsf{NP}$, and proved $\mathsf{QMA}$-hardness under polynomial-time Turing
reductions \cite{Liu2006}. Broadbent and Grilo later proved Karp hardness for
$k$-local consistency for every $k\ge5$, establishing completeness under the
standard many-one notion \cite{BroadbentGrilo2022}. Kamminga and Rudolph
subsequently sharpened the locality to every $k\ge2$
\cite{KammingaRudolph2026}. Liu, Christandl, and Verstraete obtained the
corresponding $\mathsf{QMA}$-completeness result for fermionic
$N$-representability, with hardness under Turing reductions
\cite{LiuChristandlVerstraete2007}. These are, alongside the Local Hamiltonian problem (\textsf{LH}), among the most physically natural
complete problems for $\mathsf{QMA}$. Their hardness says that
passing from a global state to its local data saves space, but not necessarily
computation.

In a way, \textsf{CLDM} is very closely related to \textsf{LH}, the problem of deciding if the lowest eigenvalue of a given local Hamiltonian is below or above some thresholds \cite{KitaevShenVyalyi2002}: \textsf{LH} optimizes a linear functional over the set of compatible local marginals, while \textsf{CLDM} asks whether a given collection belongs to this set. The two problems are therefore optimization and feasibility views of the same local data.


There are two versions of the local-to-global consistency problem. In the
standard version, the global state may be mixed. In the pure-state version, one
asks whether there is a pure state $\ket{\psi}$ whose reduced states are close
to the prescribed local marginals. Pure-state marginal problems arise naturally
in entanglement theory and quantum coding theory \cite{YuEtAl2021}. The
pure-state version is therefore both a natural analogue of \textsf{CLDM} and a
well-studied compatibility problem in its own right.

\subsection{Purity and quantum proof systems}

To decide \textsf{CLDM} in \textsf{QMA}, the verifier cuts a long witness into several blocks, treating them as several copies of a single witness. Choosing a local density matrix at random, the registers are used to estimate the corresponding reduced density matrix. In reality if the given long witness does not comprise copies of a small witness, the result is an estimation of the average of the reduced density matrices, which is the same as the reduced state of the average of the registers. This is a legitimate candidate for a global state that can be compatible or incompatible with the given local density matrices.

For an ordinary $\mathsf{QMA}$ verifier, requiring the witness to be pure makes
no difference. The acceptance probability of a probabilistic ensemble of pure states is the corresponding expectation of the acceptance probabilities of the individual pure states. A maximizer from the ensemble is both a pure state and has acceptance probability at least as large as the ensemble's. Applying it to the protocol above means we may assume the long witness is pure, but the average of its registers' reduced states might not be. Consequently the pure version of \textsf{CLDM}, \textsf{PureCLDM}, was not known, and was suspected not to be in \textsf{QMA}. 

Aharonov and Regev introduced a characterization of \textsf{QMA} that relates to the \textsf{CLDM} problem. They defined super-verifiers as families of verification
circuits with prescribed intervals for their acceptance probabilities, interpreting a family as accepting a witness if all its corresponding acceptance probabilities lie within the intervals, and rejecting if a significant number of them are significantly far
\cite{AharonovRegev2003}. Similarly to the protocol above, their simulation treats registers as copies of the witness to estimate its acceptance probability and their analysis replaces the several witness registers by their average one-register marginal, which is also a valid witness. Kamminga and Rudolph followed them and introduced the pure-witness version, \textsf{PureSuperQMA}, noting the same simulation fails for it in a similar way as the above (usual) \textsf{QMA} protocol and the pure-state \textsf{CLDM} problem \cite{KammingaRudolph2026}. The average marginal of a dishonest Merlin's witness might be mixed, so the pure state soundness promise gives no bound on it.

The difference between the pure and mixed \textsf{CLDM} problems appears already for one qubit. For each
$P\in\set{X,Y,Z}$, let $E_P=(I+P)/2$ and require acceptance probability
$1/2$. The maximally mixed state satisfies all three requirements. A pure state
has a Bloch vector $(x,y,z)$ with $x^2+y^2+z^2=1$, so it misses at least one
target by $1/(2\sqrt3)$. 

So what model \emph{can} decide \textsf{PureCLDM}? If the verifier could know in advance it got copies of a pure state in the different registers, it would be ideal. In fact, it is enough to know the witness's registers are unentangled. Under the requirement of a product witness, pure factors are as good as mixed ones: the acceptance probability of $\rho_1\otimes\cdots\otimes\rho_k$ is the convex combination of the acceptance probabilities of $\ket{\psi^1_{i_1}}\otimes\cdots\otimes\ket{\psi^k_{i_k}}$ where $\ket{\psi^j_{i_j}}$ is an eigenvector of $\rho_j$.

Kobayashi, Matsumoto, and
Yamakami introduced $\mathsf{QMA}(k)$, where Arthur receives $k$ unentangled
quantum witnesses and may measure them jointly
\cite{KobayashiMatsumotoYamakami2003}. Several classical proofs can simply be
concatenated, so the class $\mathsf{MA}(k)$ equals to $\mathsf{MA}$.
The same argument fails for quantum proofs, since a dishonest
Merlin who sends all registers at once may entangle them, while the soundness
condition of $\mathsf{QMA}(k)$ only controls product witnesses.  Thus it remains plausible that the containment $\mathsf{QMA}\subseteq\mathsf{QMA}(2)$ might be strict.

The first evidence for the usefulness of unentanglement came from protocols
with short proofs. An ordinary logarithmic-size quantum witness can be absorbed
into the verifier after length-preserving amplification \cite{MarriottWatrous2005}, which gives
$\mathsf{QMA}_{\log}=\mathsf{BQP}$ \cite{MarriottWatrous2005}. In contrast,
Blier and Tapp showed that two logarithmic-size unentangled witnesses suffice
for $\mathsf{NP}$ with an inverse-polynomial gap, using a restricted verifier
\cite{BlierTapp2012}. Aaronson et al. later gave a constant-gap protocol for
$3\text{-}\mathsf{SAT}$ using many logarithmic-size unentangled witnesses
\cite{AaronsonEtAl2009}. Chen and Drucker retained comparable parameters even
when Arthur was restricted to separate nonadaptive measurements of the
witnesses \cite{ChenDrucker2010}. This measurement restriction became the
model $\mathsf{BellQMA}$, which was shown, when the witness sizes are polynomial, to equal $\mathsf{QMA}$ \cite{GharibianSikoraUpadhyay2013,BrandaoHarrow2017}.

Harrow and Montanaro introduced the product test, which uses two copies of a
multipartite pure state to test whether it is close to a product state
\cite{HarrowMontanaro2013}. They used it to prove error reduction for multiple
Merlins and to show that two Merlins can simulate polynomially many Merlins
\[\mathsf{QMA}(2) = \mathsf{QMA}(k)\]
Aaronson et al.\ also introduced the class $\mathsf{SymQMA}(k)$.  Here the
Merlins provide identical copies of a composite witness containing all $k$
proof registers, and Arthur takes the $j$-th subsystem from the $j$-th copy, and these extracted subsystems need not be pure.
 Also, Aaronson et al.\ showed unconditionally
that a $\mathsf{SymQMA}(k)$ protocol can be simulated by a
$\mathsf{QMA}(k)$ protocol with an inverse-polynomial loss in the
completeness soundness gap, Arthur performs random SWAP tests to enforce approximate consistency among the copies before running the symmetric verification procedure \cite{AaronsonEtAl2009}. Their paper obtained the resulting class equality conditionally, since amplification for $\mathsf{QMA}(k)$ was not yet known. Harrow and Montanaro's later amplification theorem removes this loss and makes the equality unconditional \cite{HarrowMontanaro2013}. Consequently,
\[
    \mathsf{SymQMA}(k)=\mathsf{QMA}(k)=\mathsf{QMA}(2)
\]
for polynomially bounded $k\ge2$. Thus, requiring the witnesses to be identical does not by itself reduce the power of the joint-measurement model.



For quite a while there was a hope to show \textsf{PureCLDM} was complete for \textsf{QMA}(2), as it was known to be in $\mathsf{QMA}(2)$ but not in \textsf{QMA}. Kamminga and Rudolph defined the aforementioned \textsf{PureSuperQMA} to show new bounds for \textsf{PureCLDM}, which hints that it is not likely complete for $\mathsf{QMA}(2)$. Recall \textsf{PureSuperQMA} is a class defined by super-verifiers that consider only pure states \cite{KammingaRudolph2026}. In a bit more detail, a YES instance has a pure witness that satisfies every check. On a NO instance, every pure witness violates an inverse-polynomial fraction of the checks by an inverse-polynomial amount. They showed
\begin{equation}
\mathsf{QMA}
\subseteq
\mathsf{PureSuperQMA}
\subseteq
\mathsf{QMA}(2)
\notag
\end{equation}
and proved a $\mathsf{PSPACE}$ upper bound for this class. They also showed that exact pure
local consistency is complete for $\mathsf{PureSuperQMA}$ for every fixed
$k\ge2$, with analogous completeness results for bosonic and fermionic pure
$N$-representability. Motivated by the power of purity and unentangled proof, they conjectured that both containments are strict (Conjecture 1.17 \cite{KammingaRudolph2026}). They also suggested that exponential check might capture all power of $\mathsf{QMA}(2)$, $\mathsf{PureSuperQMA}(\exp) = \mathsf{QMA}(2)$, as one possible formalization from the idea that the power of $\mathsf{QMA}(2)$ derives from purity. 

In $\mathsf{QMA}(2)$, Arthur receives two unentangled witnesses that need not be
equal and he may measure them jointly \cite{KobayashiMatsumotoYamakami2003}.
Whether this model is more powerful than $\mathsf{QMA}$ remains open
\cite{AaronsonEtAl2009}. Its relation to purity comes from the fact that a pure
bipartite state is a product state exactly when its marginals are pure. This
made it plausible that $\mathsf{PureSuperQMA}$ might lie strictly between
$\mathsf{QMA}$ and $\mathsf{QMA}(2)$.

Kamminga and Rudolph also introduced
$\mathsf{BellPureSymQMA}(\poly)$. On a YES instance, an honest Merlin supplies
polynomially many copies of one pure state. Arthur measures the copies
separately and gives the classical outcomes to a quantum referee. Without the
purity restriction,
the corresponding Bell models collapse to $\mathsf{QMA}$
\cite{GharibianSikoraUpadhyay2013,BrandaoHarrow2017}. Those proofs may produce
tensor powers of mixed density operators, while Bell-pure soundness applies
only to tensor powers of pure states. Kamminga and Rudolph proved
\begin{equation}
\mathsf{PureSuperQMA}
\subseteq
\mathsf{BellPureSymQMA}(\poly)
\subseteq
\mathsf{PSPACE}
\notag
\end{equation}
but left its relation to $\mathsf{QMA}$ open
\cite{KammingaRudolph2026}.

\subsection{Our results}

We show that the common-pure-state promise gives no additional verification
power in either model. Our main theorem is
\begin{theorem*}[Informal]
\label{thm:intro-main-collapse}
\begin{equation}
\mathsf{QMA}
=
\mathsf{PureSuperQMA}
=
\mathsf{PureSuperQMA}(\mathrm{exp})
=
\mathsf{BellPureSymQMA}(\poly).
\notag
\end{equation}
\end{theorem*}
Here $\mathsf{PureSuperQMA}(\mathrm{exp})$ permits $2^{\poly(n)}$ uniformly
indexed verification checks, and its soundness condition requires every pure
state to violate an inverse-polynomial fraction of these checks by an
inverse-polynomial amount. The class does not cover regimes in which either the
violation margin or the fraction of violated checks is exponentially small. In
the Bell setting, there are polynomially many local measurements, each with
$O(\log n)$ output bits, or equivalently, a polynomial-size outcome alphabet.
Formal definitions are given in \Cref{sec:quantum-classes}.

The nontrivial directions of \Cref{thm:intro-main-collapse} are
\begin{equation}
\mathsf{PureSuperQMA}(\mathrm{exp})
\subseteq
\mathsf{QMA}
\qquad\text{and}\qquad
\mathsf{BellPureSymQMA}(\poly)
\subseteq
\mathsf{QMA}.
\notag
\end{equation}
We prove both statements by explicit uniform transformations of the
corresponding verifiers. In each case, the compiled verifier uses only a
polynomial number of witness qubits and achieves an inverse-polynomial
completeness--soundness gap, which can then be amplified to the standard
constant gap using QMA amplification \cite{MarriottWatrous2005}. The reverse
containments are immediate. They use either a single check or a single copy.
For the pure-super compiler, let $\delta$ be the fraction of checks promised to be violated by probability at least $\eps$ in the NO case. We use
$N=O\left(\delta^{-1}\epsilon^{-4}\log^2(2/\delta)\right)$ witness blocks and
obtain an unamplified gap of $\Omega(1/N)$, when there are
$M\le2^{\poly(n)}$ uniformly indexed checks and
$\epsilon,\delta\ge1/\poly(n)$. For the Bell-pure compiler, we use
$N=\Theta(m^2R)$ blocks and again obtain a gap of $\Omega(1/N)$, where
$m=\poly(n)$ is the number of local measurements and
$R=2^{\ell_{\mathrm{out}}}=\poly(n)$ is the size of the outcome alphabet. If
each witness register contains $q$ qubits, the compiled witness contains $qN$
qubits in total.

Combining the completeness theorems of Kamminga and Rudolph with our collapse
shows that $k\text{-}\mathsf{PureCLDM}_1$ is $\mathsf{QMA}$-complete for every
fixed $k\ge2$ \cite{KammingaRudolph2026}. The exact bosonic and fermionic pure
$N$-representability promise problems are $\mathsf{QMA}$-complete as well. The
mixed and pure problems have different feasible sets and very different
geometry, but they belong to the same quantum proof class. The theorem also
shows that a succinct exponential family of constraints adds no power at
inverse-polynomial margin and density. It therefore goes beyond the
polynomial-check case needed for the complete problems.

The equality $\mathsf{QMA}=\mathsf{PureSuperQMA}$ disproves the conjectured
strict containment. A general $\mathsf{QMA}(2)$ verifier receives two
unentangled states that need not be equal and can perform an arbitrary joint
measurement. Neither feature is present in the models considered here. These
common-state models therefore do not exceed $\mathsf{QMA}$. Whether the more
general $\mathsf{QMA}(2)$ model exceeds $\mathsf{QMA}$ remains open.

\subsection{Proof overview}

Both simulations ask Merlin for $N$ registers, an honest witness is
$\ket{u}^{\otimes N}$, and Arthur combines a lifted version of the source test
with a SWAP test on a uniformly random pair of registers
\cite{BuhrmanEtAl2001}. If $F_{a,b}$ swaps registers $a$ and $b$, then the
rejection effect of the SWAP branch is
\begin{equation}
\rmH_{\mathrm{asym}}
=
\frac1{\binom N2}
\sum_{1\le a<b\le N}
\frac{I-F_{a,b}}2,
\qquad
\rmH_{\mathrm{asym}}
\succeq
\frac1{N-1}
\paren*{I-\Pi_{\mathrm{sym}}^{(N)}}.
\notag
\end{equation}
(See the second relation as the random-transposition gap in
\Cref{thm:random-transposition-gap}.) Thus, the SWAP branch has the rejection effect vanishing on
$\bigvee^N(\mcH)$ and having gap of order $1/N$ on its orthogonal complement.
The lifted effects are permutation invariant, so the combined rejection
operator is block diagonal with respect to these two subspaces. It remains to
analyze the lifted test on a symmetric witness, which may still be highly
entangled.

A finite quantum de Finetti theorem would usually be the next step
\cite{ChristandlEtAl2007}. In trace norm, however, the error depends on
$\dim\mcH$, which is exponential in the size of one witness block. Results for
restricted measurements improve this dependence \cite{BrandaoHarrow2017}, but
the resulting i.i.d. components may have a mixed one-register state, so it's
not enough for a soundness promise stated only for pure tensor powers. We
therefore compare extremal test values directly, without the approximation of the
symmetric witness.

The comparison begins with a variational fact about symmetric tensors. Let
$0\ne\ket{\Psi^*}\in\bigvee^N(\mcH)$, and choose $\ket{u}$ to maximize
$\abs{\braket{\Psi^*}{u^{\otimes N}}}$. Changing $\ket u$'s phase, write
$\alpha=\braket{\Psi^*}{u^{\otimes N}}>0$.
Banach's theorem identifies this maximum with the injective norm of the
symmetric tensor \cite{Banach1938,HubenerEtAl2009}, the same quantity appears
in the geometric measure of entanglement \cite{WeiGoldbart2003}. Argmax
rounding was developed in recent work on sum-of-squares relaxations and quantum
de Finetti theorems \cite{JWXSoS,JWXdeFinetti}.

\begin{lemma*}[Informal]
For every $\ket{v}$ orthogonal to $\ket{u}$,
\begin{equation}
\braket{\Psi^*}{v\otimes u^{\otimes(N-1)}}
=0.
\notag
\end{equation}
For every $2\le k<N$ and every
$\ket{v_1},\ldots,\ket{v_k}$ orthogonal to $\ket{u}$,
\begin{equation}
\abs*{
\braket{\Psi^*}{
v_1\otimes\cdots\otimes v_k\otimes u^{\otimes(N-k)}
}}
\le
\alpha
\exp\paren*{\Oh(k)}
\paren*{\frac{k}{N}}^{k/2}
\prod_{j=1}^k\norm{v_j}_2.
\notag
\end{equation}
The constant hidden in $\Oh(k)$ is universal and independent of $\dim\mcH$.
\end{lemma*}

The estimate follows from a one-variable calculation. For a fixed
$\ket{v}\perp\ket{u}$, set
$\gamma_h(v)=\braket{\Psi^*}{v^{\otimes h}\otimes
u^{\otimes(N-h)}}$. Maximality of $\ket{u}$ gives, for every $z\in\C$,
\begin{equation}
\abs*{
\sum_{h=0}^N
\binom Nh z^h\gamma_h(v)
}
\le
\alpha
\paren*{1+\abs{z}^2\norm{v}_2^2}^{N/2}.
\notag
\end{equation}
The linear coefficient must vanish. Extracting the $k$th Fourier coefficient
on a circle of a suitable radius gives the diagonal higher-order bound, and a
polarization identity gives the stated multilinear estimate. (See the
exact bound in \Cref{thm:bosonic-argmax}.)

The factor $\alpha$ is essential. It may be exponentially small, and there is
no useful dimension-independent lower bound on it. In both applications we take $\ket{\Psi^*}$ to be
an eigenvector for the extremal eigenvalue of the lifted effect and contract
its eigenvalue equation against $\ket{u}^{\otimes N}$. The term with no
transverse factor (no $\ket v$ orthogonal to $\ket u$) is $\alpha$ times the value on a tensor power. The term with
one transverse factor vanishes, and each higher-order term carries the same
factor $\alpha$. Dividing by $\alpha$ therefore compares the extremal
symmetric value with a tensor-power value even when the maximizing overlap is
very small. This eigenvector contraction is the common analytic step in both
simulations.

For the pure-super simulation, Arthur samples a check $i$ and applies its
binary verification circuit to $t$ selected registers. The check rejects when
the empirical acceptance frequency lies outside the target interval enlarged
by $\epsilon/2$. If
$p_i(u)=\matrixel{u}{E_i}{u}$ for the accepting effect $E_i$ and $a_{i,\ell}\in\{0,1\}$ indicates whether observing $\ell$ accept outcomes causes the count test to reject, then the rejection probability on a tensor power is the Bernstein polynomial
\begin{equation}
Q_{\ol a_i}(p_i(u))
=
\sum_{\ell=0}^t
\binom t\ell
a_{i,\ell}
p_i(u)^{\ell}
\paren*{1-p_i(u)}^{t-\ell}.
\notag
\end{equation}
The polynomial is the probability of observing a rejecting count among $t$
independent Bernoulli trials, expressed in the standard Bernstein basis
\cite{Farouki2012}. Hoeffding's inequality bounds the sampling error
\cite{Hoeffding1963}. On a NO instance, the soundness promise says that every
pure state violates at least a $\delta$ fraction of the checks, and each
violated check rejects with probability close to one.
Let $\rmR_{i,t}$ be the rejecting count effect for check $i$. The averaged
effect $\rmR_t$ and its value on a tensor power are
\begin{equation}
\rmR_t
=
\frac1M\sum_{i=1}^M\rmR_{i,t},
\qquad
f_t(u)
=
\matrixel{u^{\otimes t}}{\rmR_t}{u^{\otimes t}}
=
\frac1M\sum_{i=1}^M Q_{\ol a_i}(p_i(u)).
\notag
\end{equation}

The Bernstein form also gives the transverse expansion needed for symmetric
witnesses. Write
$\ket{u_i^\perp}=\paren*{I-\ketbra{u}{u}}E_i\ket{u}$, and let
$\rmR_t^{[N]}$ be the unordered lift obtained by averaging $\rmR_t$ over the
$t$-element subsets of the $N$ registers. If $\ket{\Psi^*}$ is a least
eigenvector of $\rmR_t^{[N]}$ on $\bigvee^N(\mcH)$ with eigenvalue $\lambda$,
then
\begin{equation}
\alpha\lambda
=
\frac1M
\sum_{i=1}^M
\sum_{j=0}^t
\frac{Q_{\ol a_i}^{(j)}(p_i(u))}{j!}
\braket{\Psi^*}{
(u_i^\perp)^{\otimes j}\otimes u^{\otimes(N-j)}
}.
\notag
\end{equation}
The finite-difference identity in \Cref{lem:bernstein-derivatives} controls the
derivatives uniformly over the rejecting count predicate. Combining this
bound with the argmax estimate gives the following comparison.
\begin{equation}
0
\le
\min_{\ket{u}\in\bbS(\mcH)} f_t(u)
-
\min_{\ket{\Psi}\in\bbS(\bigvee^N(\mcH))}
\matrixel{\Psi}{\rmR_t^{[N]}}{\Psi}
\le
\Oh\paren*{\frac{t^2}{N}}.
\notag
\end{equation}

The choices in \Cref{eq:puresuper-parameters} make
$\eta_t=2e^{-t\epsilon^2/2}=\Oh(\delta)$ and
$N=\Theta(t^2/\delta)$.
On a YES instance, the tensor-power rejection is at most $\eta_t$, on a NO
instance, it is at least $\delta(1-\eta_t)$, and the stability loss leaves a
rejection probability of order $\delta$ throughout the symmetric subspace.
With $\vartheta=\Theta(1/(\delta N))$, the rejection effect is
\begin{equation}
\rmH
=
\paren*{1-\vartheta}\rmH_{\mathrm{asym}}
+
\vartheta\rmR_t^{[N]}.
\notag
\end{equation}
The count branch gives rejection $\Omega(1/N)$ on the symmetric subspace, and
the SWAP branch gives the same order on its orthogonal complement. The honest
rejection probability is $\Oh(1/N)$. Sampling one polynomial-length index
implements the average over the checks, so the argument does not enumerate an
exponentially large family.

The Bell-pure simulation uses the same contraction with a different lifted
effect. The local measurements have outcome alphabet
$\mcY=\set{0,1}^{\ell_{\mathrm{out}}(n)}$. For
$\ol a\in\mcY^m$, let $g_x(\ol a)$ be the probability that the referee accepts
the outcome tuple $\ol a$. The local POVMs and the referee then define
\begin{equation}
\rmB_x
=
\sum_{\ol a\in\mcY^m}
g_x(\ol a)
\bigotimes_{j=1}^m E_{j,a_j}.
\notag
\end{equation}
Arthur implements the ordered lift $\rmB_x^{\ang{N}}$ by assigning the $m$
labeled measurements to a uniformly random ordered choice of distinct
registers. On $\ket{u}^{\otimes N}$, its acceptance probability is exactly
$\matrixel{u^{\otimes m}}{\rmB_x}{u^{\otimes m}}$, which is the acceptance
probability of the original Bell verifier.
For every local outcome, write
$p_{u,j,a}=\matrixel{u}{E_{j,a}}{u}$ and decompose
$E_{j,a}\ket{u}=p_{u,j,a}\ket{u}+\ket{u_{j,a}^\perp}$.
POVM normalization gives the dimension-free transverse budget
$\sum_a\norm{u_{j,a}^\perp}_2\le\sqrt R$, where $R=\abs{\mcY}$ is the
size of the local outcome alphabet. Expanding the product effect, applying the
multilinear argmax estimate, and contracting an eigenvector for its largest
eigenvalue on $\bigvee^N(\mcH)$ gives the following comparison whenever $N$
is at least a sufficiently large constant times $m^2R$.
\begin{equation}
0
\le
\max_{\ket{\Psi}\in\bbS(\bigvee^N(\mcH))}
\matrixel{\Psi}{\rmB_x^{\ang{N}}}{\Psi}
-
\max_{\ket{u}\in\bbS(\mcH)}
\matrixel{u^{\otimes m}}{\rmB_x}{u^{\otimes m}}
\le
\Oh\paren*{\frac{m^2R}{N}}.
\notag
\end{equation}
The formal maximum-acceptance bound is
\Cref{thm:bosonic-stab-product}.

The short-output assumption makes $R$ polynomial. Taking
$N=\Theta(m^2R)$ with a sufficiently large constant keeps the stability loss
below the constant gap of the original Bell verifier. The Bell branch is used
with probability $1/N$, while the random-pair SWAP branch is used otherwise.
On a NO instance, the first branch gives rejection $\Omega(1/N)$ on the
symmetric subspace and the second gives the same bound on its orthogonal
complement. An honest tensor power has acceptance $1-\Oh(1/N)$ on a YES
instance. The ideal Bell compiler is proved in
\Cref{thm:bellpure-ideal-compiler}.

\subsection{Organization}

\Cref{sec:quantum-classes} gives the formal definitions and states the two
compiler theorems. \Cref{sec:bosonic-argmax} proves the bosonic argmax estimate.
\Cref{sec:count-stability,sec:puresuper-collapse} treat repeated count tests
and the pure-super compiler. \Cref{sec:product-stability,sec:bellpure-collapse}
treat product POVMs and the Bell-pure compiler. The appendix collects the
analytic and finite-sampling facts used in the proofs.


\section{Notation, Complexity Classes, and Main Results}
\label{sec:quantum-classes}
All Hilbert spaces are finite-dimensional complex spaces. The Hilbert space associated with a $q$-qubit register is $(\C^2)^{\otimes q}$, equipped with the standard computational product basis. We write
\begin{equation}
[N]:=\set{1,\ldots,N}
\qquad\text{and}\qquad
\bbS(\mcH):=\set{\ket{u}\in\mcH:\norm{u}_2=1}.
\notag
\end{equation}

Given a Hilbert space $\mcH$, an integer $N\in\N_{>0}$, and a permutation $\sigma\in\mfS_N$, define the tensor-factor permutation unitary $U_\sigma$ on $\mcH^{\otimes N}$ by
\begin{equation}
U_\sigma\ket{u_1}\otimes\cdots\otimes \ket{u_N}
:=
\ket{u_{\sigma^{-1}(1)}}\otimes\cdots\otimes \ket{u_{\sigma^{-1}(N)}}.
\notag
\end{equation} 
The symmetric subspace is
\begin{equation}
\bigvee\nolimits^N\!(\mcH)
:=
\set{
\ket{\Psi}\in\mcH^{\otimes N}:
U_\sigma\ket{\Psi}=\ket{\Psi}
\text{ for every }\sigma\in\mfS_N
}.
\notag
\end{equation}
Its orthogonal projector is
\begin{equation}
\Pi_{\mathrm{sym}}^{(N)}
:=
\frac1{N!}\sum_{\sigma\in\mfS_N}U_\sigma.
\notag
\end{equation}
An operator $A$ on $\mcH^{\otimes N}$ is permutation-invariant if
$U_\sigma^\dagger A U_\sigma=A$ for every $\sigma\in\mfS_N$. Every such operator commutes with $\Pi_{\mathrm{sym}}^{(N)}$ and is block diagonal with respect to
\begin{equation}
\mcH^{\otimes N}
=
\bigvee\nolimits^N\!(\mcH)
\oplus
\left(\bigvee\nolimits^N\!(\mcH)\right)^\perp.
\notag
\end{equation}

\subsection{Quantum Merlin--Arthur Complexity Classes}

We use the standard circuit formulation of quantum Merlin--Arthur proof systems and its witness-preserving amplification theory
\cite{KitaevShenVyalyi2002,MarriottWatrous2005}.

All verifier families are generated over a fixed finite universal gate set $\mcG$ containing the Hadamard gate and an exact universal set for reversible classical computation.To evaluate a circuit description coherently, a compiling verifier may also use the finite gate set obtained by adjoining one controlled version of each gate in $\mcG$. This extension does not change any of the classes below, a verification circuit acts on a witness register $\mcW$ and an ancilla register $\mcA$ initialized to $\ket{0}_\mcA$, and ends with a binary measurement having accepting effect $\Pi_{\mathrm{acc}}$. If $U$ is the unitary implemented before the final measurement and
\begin{equation}
J:\mcW\longrightarrow\mcW\otimes\mcA,
\qquad
J\ket{\psi}:=\ket{\psi}\ket{0}_\mcA,
\notag
\end{equation}
then the accepting effect induced on $\mcW$ is
\begin{equation}
E(U)
:=
J^\dagger U^\dagger\Pi_{\mathrm{acc}}UJ.
\notag
\end{equation}
Thus, $0\preceq E(U)\preceq I_\mcW$, and the acceptance probability on $\ket{\psi}\in\bbS(\mcW)$ is
\begin{equation}
\matrixel{\psi}{E(U)}{\psi}.
\notag
\end{equation}

\begin{definition}[$\mathsf{QMA}$]\label{def:qma}
A promise problem $L=(L_{\mathrm{yes}},L_{\mathrm{no}})$ is in $\mathsf{QMA}$ if the following holds: There are a polynomial-time computable function $q:\N\to\N_{>0}$ bounded by a polynomial and a polynomial-time uniform family of polynomial-size verification circuits
$V=(V_x)_{x\in\two^*}$. The circuit $V_x$ acts on the witness register
$\mcW_x:=(\C^2)^{\otimes q(|x|)}$ and polynomially many fresh-zero ancilla qubits. Writing $E_x:=E(V_x)$, every input $x\in\two^*$ satisfies
\begin{align}
x\in L_{\mathrm{yes}}
&\Longrightarrow
\paren[\Big]{\exists\ket{\psi}\in\bbS(\mcW_x)}
\brak[\Big]{\matrixel{\psi}{E_x}{\psi}\ge\frac23},
\notag \\
x\in L_{\mathrm{no}}
&\Longrightarrow
\paren[\Big]{\forall\ket{\psi}\in\bbS(\mcW_x)}
\brak[\Big]{\matrixel{\psi}{E_x}{\psi}\le\frac13}.
\notag
\end{align}
Quantifying over unit vectors is equivalent to quantifying over density operators because the acceptance probability is linear and its maximum is the largest eigenvalue of $E_x$. \hfill $\lozenge$
\end{definition}

We will also use the standard equivalent definition in which the completeness and soundness thresholds have an inverse-polynomial gap. Witness-preserving amplification of \cite{MarriottWatrous2005} converts such a verifier to the constant thresholds ($\tfrac23$ and $\tfrac13$ resp.) in \Cref{def:qma}.

The super-verifier formalism was introduced by Aharonov and Regev in
their definition of $\mathsf{QMA}^{+}$ \cite{AharonovRegev2003}.
The pure-state class used below, $\mathsf{PureSuperQMA}$, was introduced by Kamminga and Rudolph \cite{KammingaRudolph2026}.

\begin{definition}[Pure super-verifier]\label{def:puresuper-verifier}
Let $M:\N\to\N_{>0}$ be such that the binary representation of $M(n)$ can be computed from $1^n$ in time polynomial in $n$. A uniform $M$-check pure super-verifier
$V$ consists of a polynomially bounded, polynomial-time computable witness-length function $q:\N\to\N_{>0}$ and a family of triples
\begin{equation}
\set{
(V_{x,i},r_{x,i},s_{x,i}):
x\in\two^*,\ i\in[M(|x|)]
}
\notag
\end{equation}
with the following properties: The circuit $V_{x,i}$ acts on $q(|x|)$ witness qubits and polynomially many fresh-zero ancilla qubits, and has polynomial size. The numbers
$r_{x,i},s_{x,i}\in[0,1]$ are rationals having polynomial bit length. A deterministic generator outputs
$(V_{x,i},r_{x,i},s_{x,i})$ from $(x,i)$ in time polynomial in
$|x|+\log M(|x|)$. Writing $E_{x,i}:=E(V_{x,i})$, we use
\begin{equation}
\operatorname{acc}_V(x,i,\psi)
:=
\matrixel{\psi}{E_{x,i}}{\psi}
\notag
\end{equation}
for the acceptance probability on
$\ket{\psi}\in\bbS((\C^2)^{\otimes q(|x|)})$. \hfill $\lozenge$
\end{definition}

\begin{definition}[$\mathsf{PureSuperQMA}$]\label{def:puresuperqma}
A promise problem $L=(L_{\mathrm{yes}},L_{\mathrm{no}})$ is in $\mathsf{PureSuperQMA}(\mathrm{exp})$ if there are 
positive integer-valued
polynomials $b,c,d$, a function
$M:\N\to\N_{>0}$ satisfying $M(n)\le2^{d(n)}$, and a uniform $M$-check pure super-verifier $V$ such that, for every input
$x\in\two^*$ of length $n$,
\begin{align}
x\in L_{\mathrm{yes}}
&\Longrightarrow
\paren[\Big]{\exists\ket{\psi}\in\bbS((\C^2)^{\otimes q(n)})}
\paren[\Big]{\forall i\in[M(n)]}
\brak[\Big]{
\abs{\operatorname{acc}_V(x,i,\psi)-r_{x,i}}
\le s_{x,i}
},
\notag \\
x\in L_{\mathrm{no}}
&\Longrightarrow
\paren[\Big]{\forall\ket{\psi}\in\bbS((\C^2)^{\otimes q(n)})}
\brak[\Big]{
\Pr_{i\sim\operatorname{Unif}([M(n)])}
\brak*{
\abs{\operatorname{acc}_V(x,i,\psi)-r_{x,i}}
>
s_{x,i}+\frac1{b(n)}
}
\ge\frac1{c(n)}
}.
\notag
\end{align}
The class $\mathsf{PureSuperQMA}(\poly)$ is defined by replacing
$M(n)\le2^{d(n)}$ with $M(n)\le d(n)$. We also write
$\mathsf{PureSuperQMA}:=\mathsf{PureSuperQMA}(\poly)$. 
\hfill $\lozenge$
\end{definition}

Next we define a version of $\mathsf{QMA}$ where the witness is guaranteed to consist of several copies of a single state $\ket{\psi}^{\otimes m}$ and the verifier is restricted to measure each copy separately and decide to accept or reject by classically postprocessing the measurements' outcomes.

We use the Bell-pure symmetric model introduced by Kamminga and Rudolph \cite{KammingaRudolph2026}.

\begin{definition}[$\mathsf{BellPureSymQMA}$]\label{def:bellpuresymqma}
A uniform Bell-pure symmetric verifier consists of polynomially bounded, polynomial-time computable functions $m,q,\ell_{\mathrm{out}}:\N\to\N_{>0}$, a constant $C>0$, quantum circuits
\begin{equation}
\set{
U_{x,j}:
x\in\two^*,\ j\in[m(|x|)]
},
\notag
\end{equation}
which will be called local circuits and quantum circuits $(\mcR_x)_{x\in\two^*}$ which will be called referee circuits, with the following properties. Each $U_{x,j}$ acts on a $q(|x|)$-qubit witness register and polynomially many fresh-zero ancilla qubits (denote their number by $p_j(\abs{x})$). It has a designated output register of
$\ell_{\mathrm{out}}(|x|)$ qubits, where
\begin{equation}
\ell_{\mathrm{out}}(n)\le C\log_2(n).
\notag
\end{equation}
 $\mcR_x$\footnote{Allowing for different number of output qubits for each local circuit does not change the definition, as shorter output registers may be padded with deterministic zeros.} acts on $m(\abs{x})\cdot\ell_{\mathrm{out}(\abs{x})}$ qubits (the output qubits of all the $m(\abs{x})$ local circuits) and polynomially many fresh-zero ancilla qubits and has a designated one output qubit. 
The local circuits are generated from $(x,j)$, and the referee from $x$, in polynomial time.

The verifier acts as follows: It first acts with the $m:=m(\abs{x})$ circuits on the given witness copies $\bigotimes_{j\in[m]}(U_{x,j}\ket{\psi}\otimes\ket{0}^{\otimes p_j(\abs{x})})$. Then the output qubits are of all the local circuits are measured in the computational basis. Lastly the referee $\mcR_x$ receives the ordered tuple of $m(|x|)$ measurement outcomes, acts on them with polynomially many fresh-zero ancilla qubits, and its output qubit is measured in the computational basis, giving the verifier's answer. 

For $x\in\two^*$, put $n:=|x|$. For $\ket{\psi}\in\bbS((\C^2)^{\otimes q(n)})$, let
$\beta_x(\psi)$ be the acceptance probability obtained by applying
$U_{x,j}$ to the $j$th factor of $\ket{\psi}^{\otimes m(n)}$, measuring every designated output, and running $\mcR_x$ on the ordered outcome tuple. A promise problem
$L=(L_{\mathrm{yes}},L_{\mathrm{no}})$ is in
$\mathsf{BellPureSymQMA}(\poly)$ if it has such a verifier and every input $x$ satisfies
\begin{align}
x\in L_{\mathrm{yes}}
&\Longrightarrow
\paren[\Big]{\exists\ket{\psi}\in\bbS((\C^2)^{\otimes q(|x|)})}
\brak[\Big]{\beta_x(\psi)\ge\frac23},
\notag \\
x\in L_{\mathrm{no}}
&\Longrightarrow
\paren[\Big]{\forall\ket{\psi}\in\bbS((\C^2)^{\otimes q(|x|)})}
\brak[\Big]{\beta_x(\psi)\le\frac13}.
\notag
\end{align}
\hfill $\lozenge$
\end{definition}

\begin{remark}
    The existential quantifier in the $x\in L_{\mathrm{yes}}$ case above is intentional, the corresponding definition in the full version of Kamminga and Rudolph's work \cite[Definition~8.12]{KammingaRudolph2026} instead prints a universal quantifier over $\ket{\psi}$, but both its description of Merlin's proof and its asserted containment

    $\mathsf{PureSuperQMA}\subseteq\mathsf{BellPureSymQMA}(\poly)$ require the existential quantifier. We therefore regard the printed universal quantifier as a typographical error. 
\end{remark}

\subsection{Main Results}

\begin{theorem}\label{thm:puresuper-compiler}
Let $b,c,d,M,V$ witness that
$L\in\mathsf{PureSuperQMA}(\mathrm{exp})$ as in
\Cref{def:puresuperqma}. On inputs of length $n$, set
$\epsilon=1/b(n)$, $\delta=1/c(n)$, and $q=q(n)$. There is a uniform
one-witness verifier using
\begin{equation}
\ell_\delta
:=
\min\set{\ell\in\N_{\ge0}:2^\ell\delta\ge32},
\qquad
t:=\frac{2\ell_\delta}{\epsilon^2},
\qquad
N:=\frac{256t^2}{\delta}.
\notag
\end{equation}
There is a uniform one-witness verifier using $qN$ witness qubits whose completeness and soundness satisfy
\begin{equation}
c_{\mathrm{PS}}
\ge
1-\frac{17}{1024N},
\qquad
s_{\mathrm{PS}}
\le
1-\frac{223}{1024N}.
\notag
\end{equation}
In particular,
\begin{equation}
c_{\mathrm{PS}}-s_{\mathrm{PS}}
\ge
\frac{103}{512N}.
\notag
\end{equation}
\end{theorem}

\begin{theorem}\label{thm:bellpure-compiler}
Let $V$ be a uniform Bell-pure symmetric verifier for
$L\in\mathsf{BellPureSymQMA}(\poly)$. On inputs of length $n$, set
$m=m(n)$, $q=q(n)$, and $R=2^{\ell_{\mathrm{out}}(n)}$. There is a uniform
one-witness verifier using $q\,\Oh(m^2R)$ witness qubits. Its completeness is
at least $1-\Oh((m^2R)^{-1})$, its soundness is at most
$1-\Omega((m^2R)^{-1})$, and its raw completeness--soundness gap is
$\Omega((m^2R)^{-1})$. The implicit constants are absolute and independent of
the witness-space dimension $2^q$.
\end{theorem}

\begin{theorem}\label{thm:main-collapses}
\begin{equation}
\mathsf{PureSuperQMA}(\mathrm{exp})
=
\mathsf{QMA}
\qquad\text{and}\qquad
\mathsf{BellPureSymQMA}(\poly)
=
\mathsf{QMA}.
\notag
\end{equation}
Equivalently,
\begin{equation}
\mathsf{QMA}
=
\mathsf{PureSuperQMA}
=
\mathsf{PureSuperQMA}(\mathrm{exp})
=
\mathsf{BellPureSymQMA}(\poly).
\notag
\end{equation}
\end{theorem}


\section{Dimension-Free Bosonic Argmax Stability}
\label{sec:bosonic-argmax}
Throughout, fix a Hilbert space $\mcH$, an integer $N \in \N_{\ge 2}$, and a non-zero vector $\ket{\Psi} \in \bigvee^N\!(\mcH)$. 
The argmax-rounding viewpoint used here was developed for sum-of-squares
relaxations and sharp quantum de Finetti bounds
\cite{JWXSoS,JWXdeFinetti}. More precisely,
\Cref{thm:bosonic-argmax} extends the complex-sphere argmax estimate of
\cite{JWXdeFinetti}: that estimate gives the first-order vanishing and a
sharp second-order bound, whereas the applications below require control
at every transverse order.

\begin{theorem}\label{thm:bosonic-argmax}
Let $\ket{u} \in \bbS(\mcH)$ maximize
\begin{equation}
\abs[\big]{\braket{\Psi}{u^{\otimes N}}}
\notag
\end{equation}
over $\bbS(\mcH)$, and choose its phase so that
\begin{equation}
\alpha := \braket{\Psi}{u^{\otimes N}} > 0.
\notag
\end{equation}
For every $\ket{v} \in \ker(\bra{u})$,
\begin{equation}\label{eq:bosonic-argmax-first-order}
\braket{\Psi}{v \otimes u^{\otimes (N-1)}} = 0.
\end{equation}
For every $k \in \set{2, \ldots, N-1}$ and every $\ket{v} \in \ker(\bra{u})$,
\begin{equation}\label{eq:bosonic-argmax-diagonal}
\abs[\big]{\braket{\Psi}{v^{\otimes k} \otimes u^{\otimes (N-k)}}}
\le
\alpha \paren*{\frac{ek}{N}}^{k/2}\norm{v}_2^k.
\end{equation}
Moreover, for every $k \in \set{2, \ldots, N-1}$ and every
$\ket{v_1}, \ldots, \ket{v_k} \in \ker(\bra{u})$,
\begin{align}
&\abs[\big]{\braket{\Psi}{v_1 \otimes \cdots \otimes v_k \otimes u^{\otimes (N-k)}}} \le
\alpha \frac{k^k}{k!}\paren*{\frac{ek}{N}}^{k/2}
\prod_{r=1}^k \norm{v_r}_2
\le
\alpha e^k\paren*{\frac{ek}{N}}^{k/2}
\prod_{r=1}^k \norm{v_r}_2.
\label{eq:bosonic-argmax-multilinear}
\end{align}
\end{theorem}
\begin{proof}
The map
\begin{equation}
(\ket{w_1},\ldots,\ket{w_N})
\longmapsto
\braket{\Psi}{w_1 \otimes \cdots \otimes w_N}
\notag
\end{equation}
is a non-zero symmetric multilinear form. By \Cref{lem:torus-polarization}, its diagonal restriction is not identically zero. The maximum in the statement therefore exists by compactness and is positive, and the phase of $\ket{u}$ may be chosen as stated.
Fix $\ket{v} \in \ker(\bra{u})$. For each $h \in \set{0,\ldots,N}$, write
\begin{equation}
\gamma_h(v) := \braket{\Psi}{v^{\otimes h} \otimes u^{\otimes (N-h)}}.
\notag
\end{equation}
For every $z \in \C$, maximality of $\ket{u}$ gives
\begin{align}
\abs*{\bra{\Psi}\paren*{\ket{u}+z\ket{v}}^{\otimes N}}
&=
\abs*{\sum_{h=0}^N \binom{N}{h}z^h\gamma_h(v)}
\notag \le
\alpha\paren*{1+\abs{z}^2\norm{v}_2^2}^{N/2}.
\notag
\end{align}
Suppose for contradiction that $\gamma_1(v)\ne 0$, and for $r>0$ set
\begin{equation}
z_r := r\exp\paren[\big]{-i\,\arg(\gamma_1(v))}.
\notag
\end{equation}
After squaring the preceding inequality, subtracting $\alpha^2$, dividing by $r$, and taking $r \to 0^+$, we obtain
\begin{align}
0 < 2N\alpha\abs{\gamma_1(v)}
&=
\lim_{r\to 0^+}
\frac{
\abs*{\sum_{h=0}^N\binom{N}{h}z_r^h\gamma_h(v)}^2-\alpha^2
}{r}
\notag \\
&\le
\lim_{r\to 0^+}
\alpha^2\frac{
\paren*{1+r^2\norm{v}_2^2}^{N}-1
}{r}
=0,
\notag
\end{align}
which is impossible. Thus $\gamma_1(v)=0$, proving \Cref{eq:bosonic-argmax-first-order}.
Now, suppose that $\ket{v}$ is a unit vector, and fix $k \in \set{2,\ldots,N-1}$. For every $r>0$ and $\theta \in [0,2\pi]$, Fourier inversion and the same maximality inequality give
\begin{align}
\binom{N}{k}r^k\abs{\gamma_k(v)}
&=
\abs*{
\frac{1}{2\pi}
\int_0^{2\pi}
e^{-ik\theta}
\bra{\Psi}\paren*{\ket{u}+re^{i\theta}\ket{v}}^{\otimes N}
\,\mathrm{d}\theta
} \le
\alpha\paren*{1+r^2}^{N/2}.
\notag
\end{align}
Taking $r=\sqrt{k/(N-k)}$ and using
\begin{equation}
\binom{N}{k}
=\prod_{h=0}^{k-1}\frac{N-h}{k-h}
\ge \paren*{\frac{N}{k}}^k,
\notag
\end{equation}
we obtain
\begin{align}
\frac{\paren*{1+r^2}^{N/2}}{\binom{N}{k}r^k}
&\le
\frac{
\paren*{\frac{N}{N-k}}^{N/2}
}{
\paren*{\frac{N}{k}}^k
\paren*{\frac{k}{N-k}}^{k/2}
} =
\paren*{\frac{N}{N-k}}^{(N-k)/2}
\paren*{\frac{k}{N}}^{k/2}
\le
\paren*{\frac{ek}{N}}^{k/2}.
\notag
\end{align}
This proves \Cref{eq:bosonic-argmax-diagonal} for unit vectors, and homogeneity proves it for every $\ket{v} \in \ker(\bra{u})$.
Finally, fix $k \in \set{2,\ldots,N-1}$ and unit vectors
$\ket{v_1},\ldots,\ket{v_k} \in \ker(\bra{u})$. The map
\begin{equation}
(\ket{w_1},\ldots,\ket{w_k})
\longmapsto
\braket{\Psi}{w_1 \otimes \cdots \otimes w_k \otimes u^{\otimes (N-k)}}
\notag
\end{equation}
is a symmetric multilinear form on $\ker(\bra{u})^k$. For
$\ol{\theta}=(\theta_1,\ldots,\theta_k)\in[0,2\pi]^k$, define
\begin{equation}
\ket{w_{\ol{\theta}}}
:=
\sum_{r=1}^k e^{i\theta_r}\ket{v_r}.
\notag
\end{equation}
By \Cref{lem:torus-polarization} and \Cref{eq:bosonic-argmax-diagonal},
\begin{align}
\abs[\big]{\braket{\Psi}{v_1 \otimes \cdots \otimes v_k \otimes u^{\otimes (N-k)}}}
& \le
\frac{1}{k!(2\pi)^k}
\int_{[0,2\pi]^k}
\abs[\big]{\braket{\Psi}{w_{\ol{\theta}}^{\otimes k}\otimes u^{\otimes(N-k)}}}
\,\mathrm{d}\ol{\theta}
\notag \\
& \le
\alpha\frac{k^k}{k!}
\paren*{\frac{ek}{N}}^{k/2}.
\notag
\end{align}
Here, we used $\norm{w_{\ol{\theta}}}_2\le k$. Multilinearity gives the product of norms in \Cref{eq:bosonic-argmax-multilinear}, and \Cref{lem:elementary-estimates} gives $k^k/k!\le e^k$. This completes the proof.
\end{proof}


\section{Dimension-Free Stability of Repeated Count POVMs}\label{sec:count-stability}

Finite quantum de Finetti theorems compare marginals of symmetric states with mixtures of tensor powers \cite{ChristandlEtAl2007}. Restricted-measurement versions can improve their dimension dependence \cite{BrandaoHarrow2017}. For the repeated binary measurements used here, it is enough to compare the extremal value of a lifted count test with its value on tensor powers. \Cref{thm:bosonic-stab-count} obtains such a comparison with no dependence on $\dim\mcH$. Its proof combines the eigenvector-contraction argument of \cite{JWXdeFinetti} with the transverse expansion of a count test and standard estimates for Bernstein polynomials.

Throughout, fix a finite-dimensional Hilbert space $\mcH$ and integers
$t,N\in\N_{>0}$ such that $t\le N$.

\subsection{Count-Test Effects}

Fix an effect $E$ on $\mcH$. For every $\ell\in\set{0,\ldots,t}$, define the exact-count effect
\begin{equation}
\rmM_\ell(E)
:=
\sum_{S\in\binom{[t]}{\ell}}
\bigotimes_{h=1}^t\widetilde E_S^{(h)},
\qquad
\widetilde E_S^{(h)}
:=
\begin{cases}
E & \text{if }h\in S,\\
I_\mcH-E & \text{if }h\notin S.
\end{cases}
\notag
\end{equation}
Each $\rmM_\ell(E)$ is positive semidefinite, and
\begin{align}
\sum_{\ell=0}^t\rmM_\ell(E)
&=
\sum_{S\subseteq[t]}
\bigotimes_{h=1}^t\widetilde E_S^{(h)}
\notag \\
&=
\bigotimes_{h=1}^t\paren[\big]{E+(I_\mcH-E)}
=I_{\mcH^{\otimes t}}.
\notag
\end{align}
Thus, $\set{\rmM_\ell(E):\ell\in\set{0,\ldots,t}}$ is a POVM.
For every $\ol a=(a_0,\ldots,a_t)\in[0,1]^{t+1}$, define the count-test effect
\begin{equation}
\rmR_{\ol a}(E)
:=
\sum_{\ell=0}^t a_\ell\rmM_\ell(E)
\notag
\end{equation}
and its Bernstein polynomial $Q_{\ol a}:[0,1]\to\R$ by
\begin{equation}
Q_{\ol a}(p)
:=
\sum_{\ell=0}^t
\binom{t}{\ell}a_\ell p^\ell(1-p)^{t-\ell}.
\notag
\end{equation}

This is the standard Bernstein basis representation
\cite{Farouki2012}. Since the exact-count effects form a POVM and
$a_\ell\in[0,1]$, $\rmR_{\ol a}(E)$ is an effect, each
$\rmM_\ell(E)$ is permutation invariant, since tensor-factor permutations merely permute the subsets in $\binom{[t]}{\ell}$; hence so is $\rmR_{\ol a}(E)$.

\begin{claim}\label{clm:count-test-product-value}
For every $\ket{u}\in\bbS(\mcH)$,
\begin{equation}
\matrixel{u^{\otimes t}}{\rmR_{\ol a}(E)}{u^{\otimes t}}
=
Q_{\ol a}\paren[\big]{\matrixel{u}{E}{u}}.
\notag
\end{equation}
\end{claim}
\begin{proof}
Writing $p_u:=\matrixel{u}{E}{u}$, we compute
\begin{align}
\matrixel{u^{\otimes t}}{\rmR_{\ol a}(E)}{u^{\otimes t}}
&=
\sum_{\ell=0}^t
\sum_{S\in\binom{[t]}{\ell}}
a_\ell
\prod_{h=1}^t
\matrixel{u}{\widetilde E_S^{(h)}}{u}
\notag \\
&=
\sum_{\ell=0}^t
\binom{t}{\ell}a_\ell
p_u^\ell(1-p_u)^{t-\ell}
=Q_{\ol a}(p_u),
\notag
\end{align}
which proves the claim.
\end{proof}

\subsection{Transverse Expansion of Count-Test Effects}

Fix $\ket{u}\in\bbS(\mcH)$ and write
\begin{equation}
p_u:=\matrixel{u}{E}{u},
\qquad
\ket{u^\perp}:=
\paren[\big]{I_\mcH-\ketbra{u}{u}}E\ket{u}.
\notag
\end{equation}
Then
\begin{equation}
E\ket{u}=p_u\ket{u}+\ket{u^\perp},
\qquad
\braket{u}{u^\perp}=0.
\notag
\end{equation}
Moreover, $E^2\preceq E$ gives
\begin{align}
\norm{u^\perp}_2^2
&=
\matrixel{u}{E^2}{u}
-\matrixel{u}{E}{u}^2
\notag \\
&\le
p_u-p_u^2
\le\frac14.
\notag
\end{align}
In particular, $\norm{u^\perp}_2\le\frac12$.

\begin{claim}\label{clm:count-test-action}
For every $\ol a\in[0,1]^{t+1}$,
\begin{equation}
\rmR_{\ol a}(E)\ket{u^{\otimes t}}
=
\sum_{j=0}^t
\frac{Q_{\ol a}^{(j)}(p_u)}{(t)_j}
\sum_{S\in\binom{[t]}{j}}
\bigotimes_{h=1}^t
\ket{\widetilde u_S^{(h)}},
\qquad \qquad
\ket{\widetilde u_S^{(h)}}
:=
\begin{cases}
\ket{u^\perp} & \text{if }h\in S,\\
\ket{u} & \text{if }h\notin S.
\end{cases}
\notag
\end{equation}
where $(t)_j$ is the falling factorial and $Q_{\ol{a}}^{(j)}$ is the $j$-th derivative of $Q_{\ol{a}}$. 
\end{claim}
\begin{proof}
For every $S\subseteq[t]$,
\begin{equation}
\widetilde E_S^{(h)}\ket{u}
=
\begin{cases}
p_u\ket{u}+\ket{u^\perp} & \text{if }h\in S,\\
(1-p_u)\ket{u}-\ket{u^\perp} & \text{if }h\notin S.
\end{cases}
\notag
\end{equation}
Fix $T\subseteq[t]$ and let $j:=|T|$. In the expansion of
$\rmR_{\ol a}(E)\ket{u^{\otimes t}}$, the coefficient of the tensor having $\ket{u^\perp}$ in the positions in $T$ and $\ket{u}$ elsewhere is
\begin{align}
\sum_{S\subseteq[t]}
a_{|S|}
p_u^{|S\setminus T|}
(1-p_u)^{t-|S\cup T|}
(-1)^{|T\setminus S|}
& =
\sum_{\ell=0}^{t-j}
\binom{t-j}{\ell}
p_u^\ell(1-p_u)^{t-j-\ell}
\sum_{k=0}^j
(-1)^{j-k}\binom{j}{k}a_{\ell+k}
\notag \\
&  =
\sum_{\ell=0}^{t-j}
(\Delta^j\ol a)_\ell
\binom{t-j}{\ell}
p_u^\ell(1-p_u)^{t-j-\ell}
\notag \\ & =
\frac{Q_{\ol a}^{(j)}(p_u)}{(t)_j}.
\notag
\end{align}
The final equality is the Bernstein derivative identity in \Cref{lem:bernstein-derivatives}. The coefficient depends only on $|T|=j$, which proves the claim.
\end{proof}

\subsection{Bosonic Stability of Count-Test Effects}

Let $(\Omega,\mu)$ be a finite probability space. For every $\omega\in\Omega$, let $E_\omega$ be an effect on $\mcH$ and let
$\ol a_\omega=(a_{\omega,0},\ldots,a_{\omega,t})\in[0,1]^{t+1}$. Define
\begin{equation}
\rmR_\mu
:=
\EE_{\omega\sim\mu}
\brak[\big]{\rmR_{\ol a_\omega}(E_\omega)}.
\notag
\end{equation}
For each $T\in\binom{[N]}{t}$, let $(\rmR_\mu)_T$ denote the operator on $\mcH^{\otimes N}$ given by $\rmR_\mu$ applied to the tensor factors in $T$ in increasing order and the identity applied to all remaining factors. Define the unordered lift
\begin{equation}
\rmR_\mu^{[N]}
:=
\frac{1}{\binom Nt}
\sum_{T\in\binom{[N]}{t}}
(\rmR_\mu)_T.
\notag
\end{equation}
This is a permutation-invariant effect.

\begin{claim}\label{clm:count-test-lift-value}
For every $\ket{u}\in\bbS(\mcH)$,
\begin{align}
\matrixel{u^{\otimes N}}{\rmR_\mu^{[N]}}{u^{\otimes N}}
&=
\matrixel{u^{\otimes t}}{\rmR_\mu}{u^{\otimes t}} =
\EE_{\omega\sim\mu}
\brak[\big]{
Q_{\ol a_\omega}\paren[\big]{\matrixel{u}{E_\omega}{u}}
}.
\notag
\end{align}
For every $\ket{\Phi}\in\bigvee^N\!(\mcH)$,
\begin{equation}
\matrixel{\Phi}{\rmR_\mu^{[N]}}{u^{\otimes N}}
=
\bra{\Phi}
\paren[\big]{\rmR_\mu\otimes I_\mcH^{\otimes(N-t)}}
\ket{u^{\otimes N}}.
\notag
\end{equation}
\end{claim}
\begin{proof}
Every placement in the unordered lift gives the same expectation on $\ket{u^{\otimes N}}$ and, because $\ket{\Phi}$ is permutation-invariant, the same contraction against $\bra{\Phi}$. The first assertion then follows from \Cref{clm:count-test-product-value}, and the second follows directly from the definition of the lift.
\end{proof}

\begin{theorem}\label{thm:bosonic-stab-count}
If $N\ge2e^3t^2$, then
\begin{align}
0
&\le
\min_{\ket{u}\in\bbS(\mcH)}
\EE_{\omega\sim\mu}
\brak[\big]{
Q_{\ol a_\omega}\paren[\big]{\matrixel{u}{E_\omega}{u}}
}
-
\min_{\ket{\Psi}\in\bbS(\bigvee^N\!(\mcH))}
\matrixel{\Psi}{\rmR_\mu^{[N]}}{\Psi} \le
\frac{e^3t^2}{2N}.
\notag
\end{align}
\end{theorem}
\begin{proof}
For every $\ket{u}\in\bbS(\mcH)$, the vector $\ket{u^{\otimes N}}$ lies in $\bbS(\bigvee^N\!(\mcH))$, and \Cref{clm:count-test-lift-value} gives
\begin{equation}
\matrixel{u^{\otimes N}}{\rmR_\mu^{[N]}}{u^{\otimes N}}
=
\EE_{\omega\sim\mu}
\brak[\big]{
Q_{\ol a_\omega}\paren[\big]{\matrixel{u}{E_\omega}{u}}
}.
\notag
\end{equation}
Thus, the symmetric minimum is at most the tensor-power minimum, which proves the lower bound.
For the upper bound, choose a vector
$\ket{\Psi^*}\in\bbS(\bigvee^N\!(\mcH))$ attaining the symmetric minimum and write
\begin{equation}
\rmR_\mu^{[N]}\ket{\Psi^*}
=
\lambda\ket{\Psi^*}.
\notag
\end{equation}
Let $\ket{u}\in\bbS(\mcH)$ maximize
$\abs{\braket{\Psi^*}{u^{\otimes N}}}$, and choose its phase so that
\begin{equation}
\alpha:=\braket{\Psi^*}{u^{\otimes N}}>0.
\notag
\end{equation}
For every $\omega\in\Omega$, write
\begin{equation}
p_{u,\omega}:=\matrixel{u}{E_\omega}{u},
\qquad
\ket{u_\omega^\perp}
:=
\paren[\big]{I_\mcH-\ketbra{u}{u}}E_\omega\ket{u}.
\notag
\end{equation}
Then $\braket{u}{u_\omega^\perp}=0$ and
$\norm{u_\omega^\perp}_2\le\frac12$. By self-adjointness, \Cref{clm:count-test-lift-value}, and \Cref{clm:count-test-action},
\begin{align}
\alpha\lambda
&=
\matrixel{\Psi^*}{\rmR_\mu^{[N]}}{u^{\otimes N}} =
\EE_{\omega\sim\mu}
\sum_{j=0}^t
\frac{Q_{\ol a_\omega}^{(j)}(p_{u,\omega})}{j!}
\braket{\Psi^*}{(u_\omega^\perp)^{\otimes j}\otimes u^{\otimes(N-j)}}.
\label{eq:count-stability-expansion}
\end{align}
Here, the factor $1/j!$ follows from
$(t)_j=j!\binom tj$ and the permutation invariance of $\ket{\Psi^*}$. The $j=0$ term of \Cref{eq:count-stability-expansion} equals $ \alpha \EE_{\omega\sim\mu} \brak[\big]{Q_{\ol a_\omega}(p_{u,\omega})}$
and the $j=1$ term vanishes by \Cref{eq:bosonic-argmax-first-order}. 

\noindent
For every $j\in\set{2,\ldots,t}$, \Cref{lem:bernstein-derivatives} gives
\begin{equation}
\abs{Q_{\ol a_\omega}^{(j)}(p_{u,\omega})}
\le
2^{j-1}(t)_j.
\notag
\end{equation}
Moreover, $N\ge2e^3t^2$ implies $t<N$, so \Cref{eq:bosonic-argmax-diagonal} and
$\norm{u_\omega^\perp}_2\le\frac12$ give
\begin{align}
\frac1\alpha
\abs*{
\frac{Q_{\ol a_\omega}^{(j)}(p_{u,\omega})}{j!}
\braket{\Psi^*}{(u_\omega^\perp)^{\otimes j}\otimes u^{\otimes(N-j)}}
} \le
\frac{2^{j-1}(t)_j}{j!}
\paren*{\frac12}^j
\paren*{\frac{ej}{N}}^{j/2}
=
\frac12\binom tj
\paren*{\frac{ej}{N}}^{j/2}.
\notag
\end{align}
Taking the lower triangle inequality in \Cref{eq:count-stability-expansion}, we obtain
\begin{equation}
\lambda
\ge
\EE_{\omega\sim\mu}
\brak[\big]{Q_{\ol a_\omega}(p_{u,\omega})}
-
\frac12\sum_{j=2}^t
\binom tj
\paren*{\frac{ej}{N}}^{j/2}.
\notag
\end{equation}
For every $j\in\set{2,\ldots,t}$, \Cref{lem:elementary-estimates} gives
\begin{equation}
\binom tj\paren*{\frac{ej}{N}}^{j/2}
\le
\paren*{\frac{e^{3/2}t}{\sqrt{jN}}}^j
\le
\paren*{\frac{e^{3/2}t}{\sqrt{2N}}}^j.
\notag
\end{equation}
Define
\begin{equation}
\rho:=\frac{e^{3/2}t}{\sqrt{2N}}.
\notag
\end{equation}
The assumption $N\ge2e^3t^2$ gives $\rho\le\frac12$, and hence
\begin{equation}
\frac12\sum_{j=2}^t\rho^j
\le
\frac{\rho^2}{2(1-\rho)}
\le
\rho^2
=
\frac{e^3t^2}{2N}.
\notag
\end{equation}
Consequently,
\begin{align}
\lambda
&\ge
\EE_{\omega\sim\mu}
\brak[\big]{
Q_{\ol a_\omega}\paren[\big]{\matrixel{u}{E_\omega}{u}}
}
-\frac{e^3t^2}{2N} \ge
\min_{\ket{z}\in\bbS(\mcH)}
\EE_{\omega\sim\mu}
\brak[\big]{
Q_{\ol a_\omega}\paren[\big]{\matrixel{z}{E_\omega}{z}}
}
-\frac{e^3t^2}{2N}.
\notag
\end{align}
This proves the upper bound.
\end{proof}


\section{Simulating \texorpdfstring{$\mathsf{PureSuperQMA}(\mathrm{exp})$}{PureSuperQMA(exp)} in \texorpdfstring{$\mathsf{QMA}$}{QMA}}\label{sec:puresuper-collapse}

Let $b,c,d,M,V$ witness that
$L\in\mathsf{PureSuperQMA}(\mathrm{exp})$ as in
\Cref{def:puresuperqma}. Fix a promised input $x\in\two^*$ of length $n$, and write
\begin{equation}
\mcH:=(\C^2)^{\otimes q(n)},
\qquad
M:=M(n),
\qquad
\epsilon:=\frac1{b(n)},
\qquad
\delta:=\frac1{c(n)}.
\notag
\end{equation}
For each $i\in[M]$, let $E_i$ be the accepting effect induced by
$V_{x,i}$ on $\mcH$. Thus,
\begin{equation}
0\preceq E_i\preceq I_\mcH,
\qquad
p_i(u):=\matrixel{u}{E_i}{u}
=
\operatorname{acc}_V(x,i,u)
\notag
\end{equation}
for every $\ket{u}\in\bbS(\mcH)$.

The verifier constructed below has two branches. One branch tests a randomly
chosen acceptance window on several witness registers, the other branch
checks whether a randomly chosen pair of registers is symmetric. The first
branch supplies soundness on the symmetric subspace through
\Cref{thm:bosonic-stab-count}, while the second supplies soundness on its
orthogonal complement through \Cref{thm:random-transposition-gap}.

\subsection{Acceptance Windows as Count-Test Energies}

Define
\begin{equation}\label{eq:puresuper-parameters}
\ell_\delta
:=
\min\set{\ell\in\N_{\ge0}:2^\ell\delta\ge32},
\qquad
t:=\frac{2\ell_\delta}{\epsilon^2},
\qquad
N:=\frac{256t^2}{\delta},
\qquad
\vartheta:=\frac1{4\delta N}.
\end{equation}
Because $\epsilon^{-1}=b(n)$ and $\delta^{-1}=c(n)$ are positive
integers, both $t$ and $N$ are positive integers. For later use, set
\begin{equation}\label{eq:puresuper-concentration-parameter}
\eta_t:=2e^{-t\epsilon^2/2}.
\end{equation}

For every $i\in[M]$ and $\ell\in\set{0,\ldots,t}$, define
\begin{equation}\label{eq:puresuper-count-predicate}
a_{i,\ell}
:=
\Ind\set*{
\abs*{\frac{\ell}{t}-r_{x,i}}
>
s_{x,i}+\frac\epsilon2
},
\qquad
\ol a_i:=(a_{i,0},\ldots,a_{i,t}).
\end{equation}
The definition remains meaningful when
$s_{x,i}+\epsilon/2>1$; in that case all of the indicators simply
vanish. Using the count-test notation from
\Cref{sec:count-stability}, let
\begin{equation}\label{eq:puresuper-averaged-count-effect}
\rmR_{i,t}:=\rmR_{\ol a_i}(E_i),
\qquad
\rmR_t:=\frac1M\sum_{i=1}^M\rmR_{i,t},
\qquad
f_t(u):=\matrixel{u^{\otimes t}}{\rmR_t}{u^{\otimes t}}.
\end{equation}
Each $\rmR_{i,t}$ is a permutation-invariant effect on
$\mcH^{\otimes t}$, so the same is true of $\rmR_t$. Operationally,
$\rmR_t$ is implemented by sampling $i$ uniformly and applying the
corresponding count test; the verifier never constructs or enumerates the
average in \Cref{eq:puresuper-averaged-count-effect}.

\begin{claim}\label{clm:puresuper-tensor-separation}
If $\ket{u}\in\bbS(\mcH)$ satisfies
\begin{equation}
\abs{p_i(u)-r_{x,i}}\le s_{x,i}
\notag
\end{equation}
for every $i\in[M]$, then
\begin{equation}
f_t(u)\le\eta_t.
\notag
\end{equation}
If $x\in L_{\mathrm{no}}$, then every
$\ket{u}\in\bbS(\mcH)$ satisfies
\begin{equation}
f_t(u)\ge\delta(1-\eta_t).
\notag
\end{equation}
\end{claim}
\begin{proof}
Fix $i\in[M]$ and $\ket{u}\in\bbS(\mcH)$. Apply the accepting
measurement of $V_{x,i}$ separately to the $t$ tensor factors of
$\ket{u}^{\otimes t}$, and denote the outcome bits by
$Y_1,\ldots,Y_t$. The tensor-product form of the state implies that the
outcomes are independent Bernoulli random variables with common mean
$p_i(u)$. Write
\begin{equation}
\ol Y:=\frac1t\sum_{h=1}^tY_h.
\notag
\end{equation}
By \Cref{clm:count-test-product-value} and the definition of the
predicate in \Cref{eq:puresuper-count-predicate},
\begin{equation}\label{eq:puresuper-single-check-probability}
\matrixel{u^{\otimes t}}{\rmR_{i,t}}{u^{\otimes t}}
=
\Pr\brak*{
\abs{\ol Y-r_{x,i}}
>
s_{x,i}+\frac\epsilon2
}.
\end{equation}
Moreover, \Cref{lem:hoeffding} gives
\begin{equation}\label{eq:puresuper-hoeffding-bound}
\Pr\brak*{
\abs{\ol Y-p_i(u)}
\ge
\frac\epsilon2
}
\le
2e^{-2t(\epsilon/2)^2}
=
\eta_t.
\end{equation}

Suppose first that the $i$th acceptance window is satisfied, so
$\abs{p_i(u)-r_{x,i}}\le s_{x,i}$. Whenever the count test rejects,
the triangle inequality gives
\begin{align}
\abs{\ol Y-p_i(u)}
&\ge
\abs{\ol Y-r_{x,i}}
-
\abs{p_i(u)-r_{x,i}}
\notag \\
&>
\paren*{s_{x,i}+\frac\epsilon2}-s_{x,i}
=
\frac\epsilon2.
\notag
\end{align}
Thus, the rejection event in
\Cref{eq:puresuper-single-check-probability} is contained in the
deviation event in \Cref{eq:puresuper-hoeffding-bound}, and
\begin{equation}
\matrixel{u^{\otimes t}}{\rmR_{i,t}}{u^{\otimes t}}
\le
\eta_t.
\notag
\end{equation}
If the same vector satisfies every acceptance window, averaging this
inequality over $i$ and using
\Cref{eq:puresuper-averaged-count-effect} proves the first assertion.

Now suppose that $x\in L_{\mathrm{no}}$. For the fixed vector
$\ket{u}$, define the set of violated checks
\begin{equation}
\mcB(u)
:=
\set*{
i\in[M]:
\abs{p_i(u)-r_{x,i}}>s_{x,i}+\epsilon
}.
\notag
\end{equation}
The soundness promise in \Cref{def:puresuperqma} says that
\begin{equation}\label{eq:puresuper-bad-check-density}
\frac{\abs{\mcB(u)}}{M}\ge\delta.
\end{equation}
Fix $i\in\mcB(u)$. If the $i$th count test does not reject, then
$\abs{\ol Y-r_{x,i}}\le s_{x,i}+\epsilon/2$, and the reverse triangle
inequality yields
\begin{align}
\abs{\ol Y-p_i(u)}
&\ge
\abs{p_i(u)-r_{x,i}}
-
\abs{\ol Y-r_{x,i}}
\notag \\
&>
\paren*{s_{x,i}+\epsilon}
-
\paren*{s_{x,i}+\frac\epsilon2}
=
\frac\epsilon2.
\notag
\end{align}
Consequently, \Cref{eq:puresuper-hoeffding-bound} shows that every
violated check rejects with probability at least $1-\eta_t$. The
remaining checks have nonnegative rejection probability, and therefore
\begin{align}
f_t(u)
&=
\frac1M
\sum_{i=1}^M
\matrixel{u^{\otimes t}}{\rmR_{i,t}}{u^{\otimes t}}
\notag \\
&\ge
\frac1M
\sum_{i\in\mcB(u)}(1-\eta_t)
\notag \\
&\ge
\delta(1-\eta_t),
\notag
\end{align}
where the last step uses
\Cref{eq:puresuper-bad-check-density}. This proves the second
assertion.
\end{proof}

Minimality in \Cref{eq:puresuper-parameters} gives
$\ell_\delta=\ceil*{\log_2(32/\delta)}$ and
$t\epsilon^2/2=\ell_\delta\ge\ln(32/\delta)$. Hence
\Cref{eq:puresuper-concentration-parameter} yields
$\eta_t\le\delta/16$. Moreover, $N=256t^2/\delta$ and
$\delta N=256t^2$, so $\vartheta=1/(1024t^2)$. The estimates needed
below are
\begin{equation}\label{eq:puresuper-basic-parameter-bounds}
\begin{aligned}
\eta_t\le\frac{\delta}{16},
\qquad
N\ge256t^2>2e^3t^2\qquad
0<\vartheta<\frac12,
\qquad
\delta(1-\eta_t)-\frac{e^3t^2}{2N}
\ge\frac{7\delta}{8}.
\end{aligned}
\end{equation}
Since $t$ is a positive integer, the bound on $N$ also implies
$N\ge t$ and $N\ge2$.
For the resource estimate, substitute $\epsilon^{-1}=b(n)$ and
$\delta^{-1}=c(n)$ into \Cref{eq:puresuper-parameters}:
\begin{equation}\label{eq:puresuper-asymptotic-resource-parameters}
\begin{aligned}
t
&=2b(n)^2\ceil*{\log_2(32c(n))}
=\Oh\paren*{b(n)^2\log_2(2c(n))},
\\
N
&=256c(n)t^2
=\Oh\paren*{c(n)b(n)^4\log_2^2(2c(n))}.
\end{aligned}
\end{equation}
Replacing $b(n)$ and $c(n)$ by $\epsilon^{-1}$ and $\delta^{-1}$ gives
the equivalent form used in the theorem statement.

\subsection{The Symmetry-and-Count Verifier}

Merlin sends an arbitrary state on $Nq(n)$ qubits, viewed as $N$
registers with Hilbert space $\mcH$. The parameter estimates in
\Cref{eq:puresuper-basic-parameter-bounds} will in particular show that
$N\ge t$ and $N\ge2$, so both random choices below are well-defined.
Arthur performs one of the following tests.

\begin{itemize}
\item With probability $1-\vartheta$, Arthur chooses a uniformly random
unordered pair of registers, performs the SWAP test\cite{BuhrmanEtAl2001}, and rejects on the
antisymmetric outcome.

\item With probability $\vartheta$, Arthur chooses $i\in[M]$ uniformly
and then chooses a uniformly random $t$-element subset of the $N$
registers. He runs $V_{x,i}$ separately on the selected registers with
fresh ancillas, obtains outcomes $Y_1,\ldots,Y_t$, and rejects exactly
when
\begin{equation}
\abs*{\frac1t\sum_{h=1}^tY_h-r_{x,i}}
>
s_{x,i}+\frac\epsilon2.
\notag
\end{equation}
\end{itemize}

For $1\le a<b\le N$, let $F_{a,b}$ swap registers $a$ and $b$. The
rejection effect of the first branch is
\begin{equation}\label{eq:random-transposition-effect}
\rmH_{\mathrm{asym}}
:=
\frac1{\binom N2}
\sum_{1\le a<b\le N}
\frac{I-F_{a,b}}2.
\end{equation}
Indeed, $(I-F_{a,b})/2$ is the orthogonal projector onto the
antisymmetric subspace of registers $a$ and $b$, and averaging over the
chosen pair gives \Cref{eq:random-transposition-effect}.

For $T\in\binom{[N]}t$, let $(\rmR_{i,t})_T$ denote
$\rmR_{i,t}$ placed on the registers indexed by $T$, in increasing
order, and the identity on the remaining registers. The rejection
effect of the second branch is
\begin{align}
\frac1{M\binom Nt}
\sum_{i=1}^M
\sum_{T\in\binom{[N]}t}
(\rmR_{i,t})_T
&=
\frac1{\binom Nt}
\sum_{T\in\binom{[N]}t}
(\rmR_t)_T
\notag \\
&=
\rmR_t^{[N]},
\label{eq:puresuper-count-branch-effect}
\end{align}
where the final expression is the unordered lift defined in
\Cref{sec:count-stability}. This calculation also makes explicit why
sampling a single check implements the average effect.

It follows from
\Cref{eq:random-transposition-effect,eq:puresuper-count-branch-effect}
that the ideal verifier has rejection effect
\begin{equation}\label{eq:puresuper-rejection-effect}
\rmH
:=
(1-\vartheta)\rmH_{\mathrm{asym}}
+
\vartheta\rmR_t^{[N]}.
\end{equation}
Both terms in \Cref{eq:puresuper-rejection-effect} are positive
semidefinite effects. Since \Cref{eq:puresuper-basic-parameter-bounds} gives
$0<\vartheta<1$, their weighted sum $\rmH$ is also an effect. Conjugating
$\rmH_{\mathrm{asym}}$ by a tensor-factor permutation merely permutes
the unordered pairs in
\Cref{eq:random-transposition-effect}. Similarly, conjugating
$\rmR_t^{[N]}$ merely permutes the subsets in
\Cref{eq:puresuper-count-branch-effect}. Thus, both effects commute
with every $U_\sigma$ and, in particular, with
$\Pi_{\mathrm{sym}}^{(N)}$. Therefore,
\begin{equation}\label{eq:puresuper-sector-decomposition}
\rmH
=
\Pi_{\mathrm{sym}}^{(N)}
\rmH
\Pi_{\mathrm{sym}}^{(N)}
+
\paren[\big]{I-\Pi_{\mathrm{sym}}^{(N)}}
\rmH
\paren[\big]{I-\Pi_{\mathrm{sym}}^{(N)}},
\end{equation}
so $\rmH$ is block diagonal with respect to
\begin{equation}
\mcH^{\otimes N}
=
\bigvee^N\!(\mcH)
\oplus
\bigvee^N\!(\mcH)^\perp.
\notag
\end{equation}
Finally, every transposition acts as the identity on
$\bigvee^N\!(\mcH)$, and hence
\begin{equation}\label{eq:puresuper-symmetry-kernel}
\left.\rmH_{\mathrm{asym}}\right|_{\bigvee^N(\mcH)}=0.
\end{equation}

\subsection{Completeness and Soundness}

\begin{theorem}\label{thm:puresuper-ideal-compiler}
The ideal verifier has completeness at least
$1-\Oh(N^{-1})$, soundness at most $1-\Omega(N^{-1})$, and
completeness--soundness gap $\Omega(N^{-1})$.
\end{theorem}
\begin{proof}
Suppose first that $x\in L_{\mathrm{yes}}$. By
\Cref{def:puresuperqma}, there is a
$\ket{u}\in\bbS(\mcH)$ satisfying every acceptance window. Merlin
sends $\ket{u}^{\otimes N}$. For every pair $a<b$,
\begin{equation}
F_{a,b}\ket{u}^{\otimes N}
=
\ket{u}^{\otimes N},
\notag
\end{equation}
so the random-pair SWAP branch accepts with certainty. On the other
branch, \Cref{clm:count-test-lift-value} gives
\begin{equation}
\matrixel{u^{\otimes N}}{\rmR_t^{[N]}}{u^{\otimes N}}
=
\matrixel{u^{\otimes t}}{\rmR_t}{u^{\otimes t}}
=
f_t(u).
\notag
\end{equation}
By \Cref{clm:puresuper-tensor-separation}, this value is at most
$\eta_t$. Hence, using
\Cref{eq:puresuper-rejection-effect,eq:puresuper-basic-parameter-bounds},
the total rejection probability is at most
\begin{equation}\label{eq:puresuper-ideal-completeness}
\matrixel{u^{\otimes N}}{\rmH}{u^{\otimes N}}
\le
\vartheta\eta_t
\le
\frac1{4\delta N}\frac{\delta}{16}
=
\frac1{64N}.
\end{equation}
The ideal completeness is therefore at least $1-1/(64N)$.

Now suppose that $x\in L_{\mathrm{no}}$. To apply
\Cref{thm:bosonic-stab-count}, take its finite probability space to be
$[M]$ with the uniform measure, and take
$E_\omega=E_i$ and $\ol a_\omega=\ol a_i$ when $\omega=i$.
Then its averaged effect is exactly $\rmR_t$, and
\Cref{clm:puresuper-tensor-separation} gives
\begin{equation}
\min_{\ket{u}\in\bbS(\mcH)} f_t(u)
\ge
\delta(1-\eta_t).
\notag
\end{equation}
The size premise of \Cref{thm:bosonic-stab-count} is
\Cref{eq:puresuper-basic-parameter-bounds}. The conclusion of that theorem
therefore implies
\begin{align}
\min_{\ket{\Psi}\in\bbS(\bigvee^N(\mcH))}
\matrixel{\Psi}{\rmR_t^{[N]}}{\Psi}
&\ge
\min_{\ket{u}\in\bbS(\mcH)}f_t(u)
-
\frac{e^3t^2}{2N}
\notag \\
&\ge
\delta(1-\eta_t)-\frac{e^3t^2}{2N}
\notag \\
&\ge
\frac{7\delta}{8},
\label{eq:puresuper-symmetric-minimum}
\end{align}
where the final step is
\Cref{eq:puresuper-basic-parameter-bounds}. Since
$\rmR_t^{[N]}$ preserves the symmetric subspace, the variational
characterization of its least eigenvalue turns
\Cref{eq:puresuper-symmetric-minimum} into the operator inequality
\begin{equation}\label{eq:puresuper-symmetric-count-bound}
\left.\rmR_t^{[N]}\right|_{\bigvee^N(\mcH)}
\succeq
\frac{7\delta}{8}
I_{\bigvee^N(\mcH)}.
\end{equation}

On the symmetric sector,
\Cref{eq:puresuper-rejection-effect,eq:puresuper-symmetry-kernel,eq:puresuper-symmetric-count-bound}
give
\begin{align}
\left.\rmH\right|_{\bigvee^N(\mcH)}
&=
\vartheta
\left.\rmR_t^{[N]}\right|_{\bigvee^N(\mcH)}
\notag \\
&\succeq
\frac1{4\delta N}\frac{7\delta}{8}
I_{\bigvee^N(\mcH)}
\notag \\
&=
\frac7{32N}I_{\bigvee^N(\mcH)}.
\label{eq:puresuper-symmetric-rejection-bound}
\end{align}
On the orthogonal sector, positivity of $\rmR_t^{[N]}$ and
\Cref{thm:random-transposition-gap} imply
\begin{align}
\left.\rmH\right|_{\bigvee^N(\mcH)^\perp}
&\succeq
\frac{1-\vartheta}{N-1}
I_{\bigvee^N(\mcH)^\perp}
\notag \\
&>
\frac1{2N}
I_{\bigvee^N(\mcH)^\perp}
\notag \\
&\succeq
\frac7{32N}
I_{\bigvee^N(\mcH)^\perp}.
\label{eq:puresuper-orthogonal-rejection-bound}
\end{align}
Here, the strict inequality follows from
\Cref{eq:puresuper-basic-parameter-bounds} and $N\ge2$, while the final
inequality is numerical.

The decomposition in \Cref{eq:puresuper-sector-decomposition} has no
cross terms. Thus,
\Cref{eq:puresuper-symmetric-rejection-bound,eq:puresuper-orthogonal-rejection-bound}
imply the full-space bound
\begin{equation}\label{eq:puresuper-global-rejection-bound}
\rmH\succeq\frac7{32N}I_{\mcH^{\otimes N}}.
\end{equation}
This operator inequality applies to arbitrary pure or mixed witnesses:
for every density operator $\rho$ on $\mcH^{\otimes N}$,
\begin{equation}
\Tr(\rmH\rho)\ge\frac7{32N}.
\notag
\end{equation}
The ideal soundness is therefore at most $1-7/(32N)$. Combining this
bound with \Cref{eq:puresuper-ideal-completeness} gives
\begin{equation}\label{eq:puresuper-ideal-gap}
\paren*{1-\frac1{64N}}
-
\paren*{1-\frac7{32N}}
=
\frac{13}{64N}.
\end{equation}
These explicit estimates imply all three asymptotic assertions in the
theorem.
\end{proof}

\subsection{Uniform Implementation and the Class Collapse}

\begin{proof}[Proof of \Cref{thm:puresuper-compiler}]
Set
\begin{equation}\label{eq:puresuper-sampling-error}
\xi:=\frac1{1024N}.
\end{equation}
The ideal verifier samples a Bernoulli branch and, depending on the
branch, a uniformly random element of one or more of the sets
\begin{equation}
[M],
\qquad
\binom{[N]}2,
\qquad
\binom{[N]}t.
\notag
\end{equation}
All ranks have polynomially many bits. And
\Cref{def:puresuperqma} gives
$\log_2M\le d(n)$, and
\begin{equation}
\log_2\binom N2\le2\log_2N,
\qquad
\log_2\binom Nt\le t\log_2N.
\notag
\end{equation}
The rational branch probability $\vartheta=1/(1024t^2)$ and the target
error $\xi$ also have polynomial bit length. Unordered pairs and
$t$-element subsets admit polynomial-time reversible unranking, so
\Cref{lem:finite-sampling} supplies a uniform implementation of all
these choices with total acceptance-probability error at most $\xi$ on
every witness.

For completeness, we describe the choice-dependent part of the circuit.
In the SWAP branch, reversible pair unranking identifies the two
$q(n)$-qubit registers, after which Arthur performs the usual
controlled-SWAP circuit and uncomputes the routing workspace. In the
count branch, reversible subset unranking identifies the selected
registers in increasing order. The generator from
\Cref{def:puresuper-verifier} is then run reversibly on $(x,i)$ to
produce $(V_{x,i},r_{x,i},s_{x,i})$. The circuit description is padded
to a common polynomial length, evaluated separately on the selected
registers by a program-controlled universal circuit with fresh
ancillas, and then uncomputed. Since
$r_{x,i},s_{x,i},\epsilon$ have polynomial bit length and the observed
count lies in $\set{0,\ldots,t}$, Arthur can decide the strict
inequality in \Cref{eq:puresuper-count-predicate} exactly by clearing
all rational denominators and performing a reversible integer
comparison. Thus no approximation beyond the sampling error in
\Cref{eq:puresuper-sampling-error} is needed.

Let $c_{\mathrm{PS}}$ and $s_{\mathrm{PS}}$ denote the completeness and
soundness of the implemented verifier. On a YES input, the witness from
\Cref{eq:puresuper-ideal-completeness} loses at most $\xi$ in
acceptance probability. On a NO input, every witness gains at most
$\xi$ in acceptance probability relative to the ideal verifier, by
\Cref{eq:puresuper-global-rejection-bound}. Therefore,
\begin{align}
c_{\mathrm{PS}}
&\ge
1-\frac1{64N}-\xi
=
1-\frac{17}{1024N},
\notag \\
s_{\mathrm{PS}}
&\le
1-\frac7{32N}+\xi
=
1-\frac{223}{1024N}.
\notag
\end{align}
Subtracting these estimates and using
\Cref{eq:puresuper-sampling-error} gives
\begin{equation}\label{eq:puresuper-implemented-gap}
c_{\mathrm{PS}}-s_{\mathrm{PS}}
\ge
\frac{13}{64N}-2\xi
=
\frac{103}{512N}.
\end{equation}

It remains to check the resource bounds. By
\Cref{eq:puresuper-asymptotic-resource-parameters}, Merlin's witness
has length
\begin{equation}\label{eq:puresuper-witness-bound}
q(n)N
=
q(n)\,
\Oh\paren*{
c(n)b(n)^4\log_2^2(2c(n))
}.
\end{equation}
This is polynomial in $n$. The random-pair branch uses
a polynomial-size reversible unranking and routing network, once a pair is addressed, its SWAP test uses $q(n)$ controlled exchanges. The count
branch executes $t$ circuits of polynomial size, and the generator for
each sampled $i$ runs in time polynomial in $n+\log M(n)$. The reversible routing,
arithmetic, unranking, and uncomputation procedures are also polynomial
in $N$, $t$, $q(n)$, and the relevant bit lengths. Consequently, both
the generated circuit size and the circuit-generation time are
polynomially bounded.

Finally, \Cref{eq:puresuper-implemented-gap} shows that the reciprocal
raw gap is $\Oh(N)$ and hence polynomially bounded. More explicitly,
\Cref{eq:puresuper-asymptotic-resource-parameters,eq:puresuper-implemented-gap}
give raw gap
$\Omega\paren*{\delta\epsilon^4/
\log_2^2\paren*{2/\delta}}$.
Standard
witness-preserving QMA amplification \cite{MarriottWatrous2005} converts this verifier to one
satisfying the constant thresholds in \Cref{def:qma}, without changing
the witness length. Together with
\Cref{eq:puresuper-asymptotic-resource-parameters,eq:puresuper-witness-bound,eq:puresuper-implemented-gap}, this
proves \Cref{thm:puresuper-compiler} and yields
\begin{equation}\label{eq:puresuper-upper-containment}
\mathsf{PureSuperQMA}(\mathrm{exp})
\subseteq
\mathsf{QMA}.
\end{equation}
\end{proof}

\begin{proposition}\label{prop:qma-in-puresuper}
It holds that
\begin{equation}
\mathsf{QMA}
\subseteq
\mathsf{PureSuperQMA}.
\notag
\end{equation}
\end{proposition}
\begin{proof}
Let $L\in\mathsf{QMA}$ and let $V=(V_x)$ satisfy
\Cref{def:qma}. Use $V_x$ as the unique pure super-check and set
\begin{equation}
M(n):=1,
\qquad
r_{x,1}:=1,
\qquad
s_{x,1}:=\frac13,
\qquad
b(n):=4,
\qquad
c(n):=1,
\qquad
d(n):=1.
\notag
\end{equation}
This is a uniform one-check family, and its witness length and circuit
size are those of the original QMA verifier.

If $x\in L_{\mathrm{yes}}$, \Cref{def:qma} provides a unit vector
$\ket{\psi}$ whose acceptance probability is at least $2/3$. Since an
acceptance probability is at most one,
\begin{equation}
\abs{\operatorname{acc}_V(x,1,\psi)-r_{x,1}}
=
1-\operatorname{acc}_V(x,1,\psi)
\le
\frac13
=
s_{x,1}.
\notag
\end{equation}
Thus the unique target window is satisfied.

If $x\in L_{\mathrm{no}}$, every unit vector $\ket{\psi}$ has
acceptance probability at most $1/3$, and hence
\begin{align}
\abs{\operatorname{acc}_V(x,1,\psi)-r_{x,1}}
&=
1-\operatorname{acc}_V(x,1,\psi)
\notag \\
&\ge
\frac23
>
s_{x,1}+\frac1{b(n)}.
\notag
\end{align}
The only check is therefore violated, so the violated-check density is
one, which equals $1/c(n)$. All conditions in
\Cref{def:puresuperqma} hold with $M(n)=1\le d(n)$, proving the
proposition.
\end{proof}

\begin{corollary}\label{cor:puresuper-collapse}
\begin{equation}
\mathsf{QMA}
=
\mathsf{PureSuperQMA}
=
\mathsf{PureSuperQMA}(\mathrm{exp}).
\notag
\end{equation}
\end{corollary}
\begin{proof}
\Cref{prop:qma-in-puresuper} gives
$\mathsf{QMA}\subseteq\mathsf{PureSuperQMA}$, while
\Cref{eq:puresuper-upper-containment} gives
$\mathsf{PureSuperQMA}(\mathrm{exp})\subseteq\mathsf{QMA}$.
The remaining containment follows from \Cref{def:puresuperqma}:
because $d(n)$ is a positive integer,
$M(n)\le d(n)\le2^{d(n)}$. Hence all three classes are equal.
\end{proof}


\section{Dimension-Free Stability of Product POVMs}\label{sec:product-stability}

Throughout, fix a finite-dimensional Hilbert space $\mcH$ and integers
$m,N,R\in\N_{>0}$ such that $m\le N$. For every $i\in[m]$, let $A_i$ be a finite non-empty set satisfying $|A_i|\le R$, and for every $S\subseteq[m]$, write
\begin{equation}
A_S:=\bigtimes_{i\in S}A_i.
\notag
\end{equation}

\subsection{Product-POVM Effects}

For every $i\in[m]$, fix a POVM
$\set{E_{i,a}:a\in A_i}$ on $\mcH$, and fix a function
$g:A_{[m]}\to[0,1]$. Define the product-POVM effect $\rmP$ on
$\mcH^{\otimes m}$ by
\begin{equation}
\rmP
:=
\sum_{\ol a\in A_{[m]}}
g(\ol a)
\bigotimes_{i=1}^m E_{i,a_i}.
\notag
\end{equation}
Indeed, $\rmP\succeq0$, and
\begin{align}
I_{\mcH^{\otimes m}}-\rmP
&=
\sum_{\ol a\in A_{[m]}}
\paren[\big]{1-g(\ol a)}
\bigotimes_{i=1}^m E_{i,a_i}
\succeq0.
\notag
\end{align}

Fix $\ket{u}\in\bbS(\mcH)$. For every $i\in[m]$ and $a\in A_i$, write
\begin{equation}
p_{u,i,a}:=\matrixel{u}{E_{i,a}}{u},
\qquad
\ket{u_{i,a}^\perp}
:=
\paren[\big]{I_\mcH-\ketbra{u}{u}}E_{i,a}\ket{u}.
\notag
\end{equation}
Then
\begin{equation}
E_{i,a}\ket{u}
=
p_{u,i,a}\ket{u}+\ket{u_{i,a}^\perp},
\qquad
\braket{u}{u_{i,a}^\perp}=0.
\notag
\end{equation}
The POVM normalization gives
\begin{equation}
\sum_{a\in A_i}p_{u,i,a}=1,
\qquad
\sum_{a\in A_i}\ket{u_{i,a}^\perp}=0.
\notag
\end{equation}
Moreover, $E_{i,a}^2\preceq E_{i,a}$ gives
\begin{align}
\sum_{a\in A_i}\norm{u_{i,a}^\perp}_2^2
&=
\sum_{a\in A_i}
\paren[\big]{
\matrixel{u}{E_{i,a}^2}{u}
-p_{u,i,a}^2
}
\notag \\
&\le
1-\sum_{a\in A_i}p_{u,i,a}^2
\le1.
\notag
\end{align}
Thus, Cauchy--Schwarz yields
\begin{equation}
\sum_{a\in A_i}\norm{u_{i,a}^\perp}_2
\le
\sqrt{|A_i|}
\sqrt{\sum_{a\in A_i}\norm{u_{i,a}^\perp}_2^2}
\le
\sqrt R.
\notag
\end{equation}

For every $S\subseteq[m]$ and $\ol a\in A_S$, define
\begin{equation}
h_{u,S}(\ol a)
:=
\sum_{\ol b\in A_{[m]\setminus S}}
g(\ol a,\ol b)
\prod_{i\in[m]\setminus S}p_{u,i,b_i},
\notag
\end{equation}
where $(\ol a,\ol b)$ denotes the tuple in $A_{[m]}$ whose restrictions to $S$ and $[m]\setminus S$ are $\ol a$ and $\ol b$, respectively. The product weights form a probability distribution, so
\begin{equation}
0\le h_{u,S}(\ol a)\le1.
\notag
\end{equation}

\begin{claim}\label{clm:product-povm-action}
For every $\ket{u}\in\bbS(\mcH)$,
\begin{equation}
\rmP\ket{u^{\otimes m}}
=
\sum_{S\subseteq[m]}
\sum_{\ol a\in A_S}
h_{u,S}(\ol a)
\bigotimes_{i=1}^m
\ket{\widetilde u_{S,\ol a}^{(i)}},
\notag
\end{equation}
where
\begin{equation}
\ket{\widetilde u_{S,\ol a}^{(i)}}
:=
\begin{cases}
\ket{u_{i,a_i}^\perp} & \text{if }i\in S,\\
\ket{u} & \text{if }i\notin S.
\end{cases}
\notag
\end{equation}
In particular,
\begin{equation}
h_{u,\varnothing}(\varnothing)
=
\matrixel{u^{\otimes m}}{\rmP}{u^{\otimes m}}.
\notag
\end{equation}
\end{claim}
\begin{proof}
Expanding each local factor gives
\begin{align}
\rmP\ket{u^{\otimes m}}
&=
\sum_{\ol c\in A_{[m]}}
g(\ol c)
\bigotimes_{i=1}^m
\paren[\big]{
p_{u,i,c_i}\ket{u}+\ket{u_{i,c_i}^\perp}
}
\notag \\
&=
\sum_{S\subseteq[m]}
\sum_{\ol c\in A_{[m]}}
g(\ol c)
\paren*{\prod_{i\in[m]\setminus S}p_{u,i,c_i}}
\bigotimes_{i=1}^m
\ket{\widetilde u_{S,\ol c_S}^{(i)}}
\notag \\
&=
\sum_{S\subseteq[m]}
\sum_{\ol a\in A_S}
h_{u,S}(\ol a)
\bigotimes_{i=1}^m
\ket{\widetilde u_{S,\ol a}^{(i)}}.
\notag
\end{align}
Taking $S=\varnothing$ proves the final assertion, so the claim follows.
\end{proof}

\subsection{The Ordered Lift}

Let $[N]^m_{\ne}$ denote the set of tuples
$\ol s=(s_1,\ldots,s_m)\in[N]^m$ having pairwise distinct entries, and write
\begin{equation}
(N)_m:=N(N-1)\cdots(N-m+1)=|[N]^m_{\ne}|.
\notag
\end{equation}
For each $\ol s\in[N]^m_{\ne}$, choose a permutation
$\sigma_{\ol s}$ of $[N]$ satisfying
$\sigma_{\ol s}(i)=s_i$ for every $i\in[m]$, and define
\begin{equation}
\rmP_{\ol s}
:=
U_{\sigma_{\ol s}}
\paren[\big]{\rmP\otimes I_\mcH^{\otimes(N-m)}}
U_{\sigma_{\ol s}}^\dagger.
\notag
\end{equation}
This definition is independent of the completion of $\sigma_{\ol s}$ outside $[m]$. Define the ordered lift
\begin{equation}
\rmP^{\ang{N}}
:=
\frac1{(N)_m}
\sum_{\ol s\in[N]^m_{\ne}}
\rmP_{\ol s}.
\notag
\end{equation}
The operator $\rmP^{\ang{N}}$ is a permutation-invariant effect.

\begin{claim}\label{clm:product-lift-value}
For every $\ket{u}\in\bbS(\mcH)$,
\begin{equation}
\matrixel{u^{\otimes N}}{\rmP^{\ang{N}}}{u^{\otimes N}}
=
\matrixel{u^{\otimes m}}{\rmP}{u^{\otimes m}}.
\notag
\end{equation}
For every $\ket{\Phi}\in\bigvee^N\!(\mcH)$,
\begin{equation}
\matrixel{\Phi}{\rmP^{\ang{N}}}{u^{\otimes N}}
=
\bra{\Phi}
\paren[\big]{\rmP\otimes I_\mcH^{\otimes(N-m)}}
\ket{u^{\otimes N}}.
\notag
\end{equation}
\end{claim}
\begin{proof}
The vectors $\ket{u^{\otimes N}}$ and $\ket{\Phi}$ are invariant under every tensor-factor permutation. Consequently, every term in the ordered lift gives the same expectation and the same contraction as the placement on the first $m$ factors. This proves both assertions.
\end{proof}

\subsection{Bosonic Stability of Product-POVM Effects}

\begin{theorem}\label{thm:bosonic-stab-product}
If $N\ge2e^5m^2R$, then
\begin{align}
0
&\le
\max_{\ket{\Psi}\in\bbS(\bigvee^N\!(\mcH))}
\matrixel{\Psi}{\rmP^{\ang{N}}}{\Psi}
-
\max_{\ket{u}\in\bbS(\mcH)}
\matrixel{u^{\otimes m}}{\rmP}{u^{\otimes m}}
\notag \\
&\le
\frac{e^5m^2R}{N}.
\notag
\end{align}
\end{theorem}
\begin{proof}
The lower bound follows from \Cref{clm:product-lift-value} and the fact that
$\ket{u^{\otimes N}}\in\bbS(\bigvee^N\!(\mcH))$.

For the upper bound, choose a vector
$\ket{\Psi^*}\in\bbS(\bigvee^N\!(\mcH))$ attaining the symmetric maximum and write
\begin{equation}
\rmP^{\ang{N}}\ket{\Psi^*}
=
\Lambda\ket{\Psi^*}.
\notag
\end{equation}
Let $\ket{u}\in\bbS(\mcH)$ maximize
$\abs{\braket{\Psi^*}{u^{\otimes N}}}$, and choose its phase so that
\begin{equation}
\alpha:=\braket{\Psi^*}{u^{\otimes N}}>0.
\notag
\end{equation}
By self-adjointness, \Cref{clm:product-lift-value}, and
\Cref{clm:product-povm-action},
\begin{align}
\alpha\Lambda
&=
\matrixel{\Psi^*}{\rmP^{\ang{N}}}{u^{\otimes N}}
\notag \\
&=
\sum_{S\subseteq[m]}
\sum_{\ol a\in A_S}
h_{u,S}(\ol a)
\bra{\Psi^*}
\paren*{
\bigotimes_{i=1}^m
\ket{\widetilde u_{S,\ol a}^{(i)}}
\otimes
\ket{u^{\otimes(N-m)}}
}.
\label{eq:product-stability-expansion}
\end{align}
The $S=\varnothing$ term of \Cref{eq:product-stability-expansion} is
\begin{equation}
\alpha\matrixel{u^{\otimes m}}{\rmP}{u^{\otimes m}},
\notag
\end{equation}
and every term with $|S|=1$ vanishes by
\Cref{eq:bosonic-argmax-first-order}.

Fix $j\in\set{2,\ldots,m}$ and
$S=\set{i_1<\cdots<i_j}\in\binom{[m]}j$. The assumption
$N\ge2e^5m^2R$ implies $m<N$, so
\Cref{eq:bosonic-argmax-multilinear} and the transverse budget give
\begin{align}
&\frac1\alpha
\abs*{
\sum_{\ol a\in A_S}
h_{u,S}(\ol a)
\braket{\Psi^*}{
u_{i_1,a_{i_1}}^\perp\otimes\cdots\otimes
u_{i_j,a_{i_j}}^\perp\otimes
u^{\otimes(N-j)}
}
}
\notag \\
&\qquad \le
e^j\paren*{\frac{ej}{N}}^{j/2}
\sum_{\ol a\in A_S}
\prod_{r=1}^j
\norm{u_{i_r,a_{i_r}}^\perp}_2
\notag \\
&\qquad =
e^j\paren*{\frac{ej}{N}}^{j/2}
\prod_{r=1}^j
\paren*{
\sum_{a\in A_{i_r}}\norm{u_{i_r,a}^\perp}_2
}
\notag \\
&\qquad \le
e^j\paren*{\frac{ejR}{N}}^{j/2}.
\notag
\end{align}
Taking the upper triangle inequality in
\Cref{eq:product-stability-expansion}, we obtain
\begin{equation}
\Lambda
\le
\matrixel{u^{\otimes m}}{\rmP}{u^{\otimes m}}
+
\sum_{j=2}^m
\binom mj
e^j\paren*{\frac{ejR}{N}}^{j/2}.
\notag
\end{equation}
For every $j\in\set{2,\ldots,m}$,
\Cref{lem:elementary-estimates} gives
\begin{equation}
\binom mj
e^j\paren*{\frac{ejR}{N}}^{j/2}
\le
\paren*{\frac{e^{5/2}m\sqrt R}{\sqrt{jN}}}^j
\le
\paren*{\frac{e^{5/2}m\sqrt R}{\sqrt{2N}}}^j.
\notag
\end{equation}
Define
\begin{equation}
\rho:=\frac{e^{5/2}m\sqrt R}{\sqrt{2N}}.
\notag
\end{equation}
The assumption $N\ge2e^5m^2R$ gives $\rho\le\frac12$, and hence
\begin{equation}
\sum_{j=2}^m\rho^j
\le
\frac{\rho^2}{1-\rho}
\le
2\rho^2
=
\frac{e^5m^2R}{N}.
\notag
\end{equation}
Consequently,
\begin{align}
\Lambda
&\le
\matrixel{u^{\otimes m}}{\rmP}{u^{\otimes m}}
+
\frac{e^5m^2R}{N} \le
\max_{\ket{z}\in\bbS(\mcH)}
\matrixel{z^{\otimes m}}{\rmP}{z^{\otimes m}}
+
\frac{e^5m^2R}{N}.
\notag
\end{align}
This proves the upper bound.
\end{proof}


\section{Simulating \texorpdfstring{$\mathsf{BellPureSymQMA}(\poly)$}{BellPureSymQMA(poly)} in \texorpdfstring{$\mathsf{QMA}$}{QMA}}\label{sec:bellpure-collapse}

Let $L\in\mathsf{BellPureSymQMA}(\poly)$, and let $V$ be a uniform Bell-pure symmetric verifier witnessing this membership. Fix a promised input $x\in\two^*$ of length $n$, and write
\begin{equation}
\mcH:=(\C^2)^{\otimes q(n)},
\qquad
m:=m(n),
\qquad
\mcY:=\set{0,1}^{\ell_{\mathrm{out}}(n)},
\qquad
R:=|\mcY|=2^{\ell_{\mathrm{out}}(n)}.
\notag
\end{equation}

\subsection{The Product-POVM Effect of a Bell Verifier}

For each $j\in[m]$, let $J_j$ append the fresh-zero ancillas used by $U_{x,j}$, and for each $a\in\mcY$, let $\Pi_{j,a}$ project the designated output register onto the computational-basis vector $\ket{a}$. Define the effect on $\mcH$
\begin{equation}
E_{j,a}
:=
J_j^\dagger U_{x,j}^\dagger\Pi_{j,a}U_{x,j}J_j.
\notag
\end{equation}
The orthogonal output projectors satisfy
\begin{equation}
E_{j,a}\succeq0,
\qquad
\sum_{a\in\mcY}E_{j,a}=I_\mcH,
\notag
\end{equation}
so $\set{E_{j,a}:a\in\mcY}$ is a POVM.

For every $\ol a=(a_1,\ldots,a_m)\in\mcY^m$, let
\begin{equation}
g_x(\ol a)
:=
\Pr\brak*{\mcR_x\text{ accepts the ordered tuple }\ol a}.
\notag
\end{equation}
Define the effect $\rmB_x$ on $\mcH^{\otimes m}$ by
\begin{equation}
\rmB_x
:=
\sum_{\ol a\in\mcY^m}
g_x(\ol a)
\bigotimes_{j=1}^mE_{j,a_j}.
\notag
\end{equation}
This is a product-POVM effect of the form considered in
\Cref{sec:product-stability}, and in particular
\begin{equation}
0\preceq\rmB_x\preceq I_{\mcH^{\otimes m}}.
\notag
\end{equation}

\begin{claim}\label{clm:bell-product-effect}
For every $\ket{u}\in\bbS(\mcH)$,
\begin{equation}
\beta_x(u)
=
\matrixel{u^{\otimes m}}{\rmB_x}{u^{\otimes m}}.
\notag
\end{equation}
For every $N\ge m$, the experiment that chooses a uniformly random ordered injection
$\phi:[m]\hookrightarrow[N]$, applies $U_{x,j}$ to register
$\phi(j)$ for every $j\in[m]$, and runs the original referee has accepting effect
$\rmB_x^{\ang{N}}$.
\end{claim}
\begin{proof}
On $\ket{u}^{\otimes m}$, the local outcomes are independent and satisfy
\begin{equation}
\Pr\brak*{\text{the }j\text{th outcome is }a_j}
=
\matrixel{u}{E_{j,a_j}}{u}.
\notag
\end{equation}
Conditioning on the ordered outcome tuple and applying the law of total probability gives
\begin{align}
\beta_x(u)
&=
\sum_{\ol a\in\mcY^m}
g_x(\ol a)
\prod_{j=1}^m
\matrixel{u}{E_{j,a_j}}{u}
\notag \\
&=
\matrixel{u^{\otimes m}}{\rmB_x}{u^{\otimes m}}.
\notag
\end{align}
For a fixed injection $\phi$, the same experiment places the labeled tensor factors of $\rmB_x$ on the ordered registers
$\phi(1),\ldots,\phi(m)$. Averaging over all injections is exactly the ordered lift from \Cref{sec:product-stability}, proving the second assertion.
\end{proof}

The expression defining $\rmB_x$ is used only to analyze the verifier. The verifier implements it by running the local circuits and the referee and never enumerates the joint outcome space $\mcY^m$.

\subsection{The One-Witness Verifier}

Let $N$ be the least power of two satisfying
\begin{equation}
N\ge4096m^2R.
\notag
\end{equation}
Minimality gives
\begin{equation}
4096m^2R
\le
N
<
8192m^2R.
\notag
\end{equation}
Since $e<3$,
\begin{equation}
2e^5m^2R
<
486m^2R
<
N,
\notag
\end{equation}
and
\begin{equation}
\frac{e^5m^2R}{N}
<
\frac{3^5}{4096}
<
\frac1{16}.
\notag
\end{equation}

Merlin sends one state on $Nq(n)$ qubits, viewed as $N$ registers with Hilbert space $\mcH$. Arthur performs one of the following tests.

\begin{itemize}
\item With probability $1-\frac1N$, Arthur chooses a uniformly random unordered pair of registers, performs the SWAP test, and rejects on the antisymmetric outcome.

\item With probability $\frac1N$, Arthur chooses a uniformly random ordered injection
$\phi:[m]\hookrightarrow[N]$. For each $j\in[m]$, he applies
$U_{x,j}$ to register $\phi(j)$ with fresh ancillas and measures its designated output. He preserves the order of the outcome tuple and runs $\mcR_x$, accepting exactly when the referee accepts.
\end{itemize}

By \Cref{clm:bell-product-effect}, the second branch has accepting effect
$\rmB_x^{\ang{N}}$. Using the random-transposition effect from
\Cref{eq:random-transposition-effect}, the rejection effect of the ideal verifier is
\begin{equation}\label{eq:bell-rejection-effect}
\rmH_{\mathrm{B}}
:=
\paren*{1-\frac1N}\rmH_{\mathrm{asym}}
+
\frac1N\paren[\big]{I-\rmB_x^{\ang{N}}}.
\end{equation}
Both summands in \Cref{eq:bell-rejection-effect} are positive semidefinite and permutation-invariant, so $\rmH_{\mathrm B}$ is block diagonal with respect to
\begin{equation}
\mcH^{\otimes N}
=
\bigvee^N\!(\mcH)
\oplus
\left(\bigvee^N\!(\mcH)\right)^\perp.
\notag
\end{equation}

\begin{theorem}\label{thm:bellpure-ideal-compiler}
The ideal verifier has completeness and soundness satisfying
\begin{equation}
c_{\mathrm{B,ideal}}
\ge
1-\frac1{3N},
\qquad
s_{\mathrm{B,ideal}}
\le
1-\frac{29}{48N}.
\notag
\end{equation}
In particular,
\begin{equation}
c_{\mathrm{B,ideal}}-s_{\mathrm{B,ideal}}
\ge
\frac{13}{48N}.
\notag
\end{equation}
\end{theorem}
\begin{proof}
Suppose first that $x\in L_{\mathrm{yes}}$, and choose
$\ket{u}\in\bbS(\mcH)$ satisfying $\beta_x(u)\ge2/3$. Merlin sends $\ket{u}^{\otimes N}$. Every SWAP test accepts, while
\Cref{clm:bell-product-effect} shows that the Bell branch rejects with probability at most $1/3$. Thus, the total rejection probability is at most
\begin{equation}
\frac1N\paren*{1-\frac23}
=
\frac1{3N},
\notag
\end{equation}
which proves the completeness bound.

Now suppose that $x\in L_{\mathrm{no}}$. By
\Cref{def:bellpuresymqma} and \Cref{clm:bell-product-effect},
\begin{equation}
\max_{\ket{u}\in\bbS(\mcH)}
\matrixel{u^{\otimes m}}{\rmB_x}{u^{\otimes m}}
\le
\frac13.
\notag
\end{equation}
Applying \Cref{thm:bosonic-stab-product} gives
\begin{align}
\left.\rmB_x^{\ang{N}}\right|_{\bigvee^N(\mcH)}
&\preceq
\paren*{
\frac13+\frac{e^5m^2R}{N}
}
I_{\bigvee^N(\mcH)}
\notag \\
&\preceq
\frac{19}{48}I_{\bigvee^N(\mcH)}.
\notag
\end{align}
Equivalently,
\begin{equation}
\left.
\paren[\big]{I-\rmB_x^{\ang{N}}}
\right|_{\bigvee^N(\mcH)}
\succeq
\frac{29}{48}I_{\bigvee^N(\mcH)}.
\notag
\end{equation}
Because $\rmH_{\mathrm{asym}}$ vanishes on the symmetric subspace,
\begin{equation}
\left.\rmH_{\mathrm B}\right|_{\bigvee^N(\mcH)}
\succeq
\frac{29}{48N}I_{\bigvee^N(\mcH)}.
\notag
\end{equation}
On the orthogonal complement, positivity of
$I-\rmB_x^{\ang{N}}$ and
\Cref{thm:random-transposition-gap} give
\begin{align}
\left.\rmH_{\mathrm B}\right|_{\bigvee^N(\mcH)^\perp}
&\succeq
\paren*{1-\frac1N}\frac1{N-1}
I_{\bigvee^N(\mcH)^\perp}
\notag \\
&=
\frac1N I_{\bigvee^N(\mcH)^\perp}
\succeq
\frac{29}{48N}I_{\bigvee^N(\mcH)^\perp}.
\notag
\end{align}
There are no cross terms between the two sectors. Therefore,
$\rmH_{\mathrm B}\succeq\frac{29}{48N}I$ on the full witness space, which proves the soundness bound. Subtracting the two thresholds proves the final assertion.
\end{proof}

\subsection{Uniformity and the Collapse}

Set
\begin{equation}
\xi_{\mathrm B}:=\frac1{1024N}.
\notag
\end{equation}
The outer branch probability $1/N$ is exactly dyadic. The remaining random choices range over sets of sizes
\begin{equation}
\binom N2
\qquad\text{and}\qquad
(N)_m.
\notag
\end{equation}
Their binary lengths are polynomial because
\begin{equation}
\log_2 (N)_m\le m\log_2N.
\notag
\end{equation}
Ordered injections can be unranked in polynomial time. By
\Cref{lem:finite-sampling}, there is a polynomial-time uniform implementation whose acceptance probability differs from that of the ideal verifier by at most $\xi_{\mathrm B}$ on every witness.

The implemented completeness and soundness therefore satisfy
\begin{equation}
c_{\mathrm B}
\ge
1-\frac1{3N}-\xi_{\mathrm B}
=
1-\frac{1027}{3072N}
\notag
\end{equation}
and
\begin{equation}
s_{\mathrm B}
\le
1-\frac{29}{48N}+\xi_{\mathrm B}
=
1-\frac{1853}{3072N}.
\notag
\end{equation}
Consequently,
\begin{equation}
c_{\mathrm B}-s_{\mathrm B}
\ge
\frac{13}{48N}-2\xi_{\mathrm B}
=
\frac{413}{1536N}.
\notag
\end{equation}

The output bound in \Cref{def:bellpuresymqma} gives
\begin{equation}
R=2^{\ell_{\mathrm{out}}(n)}
\le
(n+2)^C.
\notag
\end{equation}
Together with $N<8192m(n)^2R$, this shows that the witness length $q(n)N$, the circuit size, and the reciprocal of the raw gap are polynomially bounded. Standard amplification gives a verifier satisfying \Cref{def:qma}. This proves
\Cref{thm:bellpure-compiler} and the containment
\begin{equation}
\mathsf{BellPureSymQMA}(\poly)
\subseteq
\mathsf{QMA}.
\notag
\end{equation}

\begin{proposition}\label{prop:qma-in-bellpure}
\begin{equation}
\mathsf{QMA}
\subseteq
\mathsf{BellPureSymQMA}(\poly).
\notag
\end{equation}
\end{proposition}
\begin{proof}
Let $L\in\mathsf{QMA}$ and let $V=(V_x)$ satisfy
\Cref{def:qma}. Use the same witness-length function, set
\begin{equation}
m(n):=1,
\qquad
\ell_{\mathrm{out}}(n):=1,
\qquad
C:=1,
\notag
\end{equation}
and use $V_x$ as the unique local circuit with its accept/reject bit as the designated output. Let the referee output that bit. On a YES input, the existential condition follows from QMA completeness. On a NO input, QMA soundness holds for every unit vector. Since
\begin{equation}
1\le\log_2(n+2),
\notag
\end{equation}
this is a valid Bell-pure symmetric verifier. This proves the proposition.
\end{proof}

Combining the upper containment with \Cref{prop:qma-in-bellpure} gives
\begin{equation}
\mathsf{BellPureSymQMA}(\poly)
=
\mathsf{QMA}.
\notag
\end{equation}
Together with the equality proved in \Cref{sec:puresuper-collapse}, this completes the proof of \Cref{thm:main-collapses}.


\section*{AI Usage}

The initial high-level question and framing was human. The technical development and execution of this idea, including substantial parts of the proof development, were carried out with assistance from ChatGPT 5.6 Pro. We also used ChatGPT 5.6 Pro for editorial assistance in preparing the manuscript. The authors take full responsibility for the results and have independently verified all claims, proofs, and references. Substantial human effort was also devoted to the exposition, organization, and presentation of the results.

\section*{Acknowledgments}
The authors would like to thank Kunal Marwaha for suggesting the direction that manifested in this manuscript.

IL was funded by the European Union (ERC, EACTP, 101142020) and the Danish-Israeli Study Foundation in Memory of Josef and Regine Nachemson.
Any views, opinions, findings, and conclusions or recommendations expressed in this material are those of the author(s) and do not necessarily reflect the views of the European Union or the European Research Council Executive Agency. Neither the European Union nor any granting authority can be held responsible for them.

\printbibliography

@book{KitaevShenVyalyi2002,
  author    = {Kitaev, Alexei Yu. and Shen, Alexander H. and Vyalyi, Mikhail N.},
  title     = {Classical and Quantum Computation},
  series    = {Graduate Studies in Mathematics},
  volume    = {47},
  publisher = {American Mathematical Society},
  address   = {Providence, Rhode Island},
  year      = {2002},
  doi       = {10.1090/gsm/047},
  url       = {https://doi.org/10.1090/gsm/047}
}

@inproceedings{AharonovRegev2003,
  author    = {Aharonov, Dorit and Regev, Oded},
  title     = {A Lattice Problem in Quantum {NP}},
  booktitle = {Proceedings of the 44th Annual {IEEE} Symposium on Foundations of Computer Science},
  pages     = {210--219},
  publisher = {IEEE Computer Society},
  year      = {2003},
  doi       = {10.1109/SFCS.2003.1238195},
  url       = {https://doi.org/10.1109/SFCS.2003.1238195}
}

@article{MarriottWatrous2005,
  author  = {Marriott, Chris and Watrous, John},
  title   = {Quantum {Arthur--Merlin} Games},
  journal = {Computational Complexity},
  volume  = {14},
  number  = {2},
  pages   = {122--152},
  year    = {2005},
  doi     = {10.1007/s00037-005-0194-x},
  url     = {https://doi.org/10.1007/s00037-005-0194-x}
}

@inproceedings{KammingaRudolph2026,
  author    = {Kamminga, Jonas and Rudolph, Dorian},
  title     = {The Pure-State Consistency of Local Density Matrices Problem: In {PSPACE} and Complete for a Class Between {QMA} and {QMA(2)}},
  booktitle = {17th Innovations in Theoretical Computer Science Conference ({ITCS} 2026)},
  editor    = {Saraf, Shubhangi},
  series    = {Leibniz International Proceedings in Informatics},
  volume    = {362},
  pages     = {83:1--83:23},
  publisher = {Schloss Dagstuhl--Leibniz-Zentrum f{\"u}r Informatik},
  year      = {2026},
  doi       = {10.4230/LIPIcs.ITCS.2026.83},
  url       = {https://doi.org/10.4230/LIPIcs.ITCS.2026.83},
  note      = {Full version available as arXiv:2411.03096}
}

@article{Coleman1963,
  author  = {Coleman, A. J.},
  title   = {Structure of Fermion Density Matrices},
  journal = {Reviews of Modern Physics},
  volume  = {35},
  number  = {3},
  pages   = {668--686},
  year    = {1963},
  doi     = {10.1103/RevModPhys.35.668},
  url     = {https://doi.org/10.1103/RevModPhys.35.668}
}

@article{Klyachko2006,
  author  = {Klyachko, Alexander A.},
  title   = {Quantum Marginal Problem and {$N$}-Representability},
  journal = {Journal of Physics: Conference Series},
  volume  = {36},
  pages   = {72--86},
  year    = {2006},
  doi     = {10.1088/1742-6596/36/1/014},
  url     = {https://doi.org/10.1088/1742-6596/36/1/014}
}

@inproceedings{Liu2006,
  author    = {Liu, Yi-Kai},
  title     = {Consistency of Local Density Matrices Is {QMA}-Complete},
  booktitle = {Approximation, Randomization, and Combinatorial Optimization. Algorithms and Techniques},
  series    = {Lecture Notes in Computer Science},
  volume    = {4110},
  pages     = {438--449},
  publisher = {Springer},
  year      = {2006},
  doi       = {10.1007/11830924_40},
  url       = {https://doi.org/10.1007/11830924_40}
}

@article{BroadbentGrilo2022,
  author  = {Broadbent, Anne and Grilo, Alex Bredariol},
  title   = {{QMA}-Hardness of Consistency of Local Density Matrices with Applications to Quantum Zero-Knowledge},
  journal = {SIAM Journal on Computing},
  volume  = {51},
  number  = {4},
  pages   = {1400--1450},
  year    = {2022},
  doi     = {10.1137/21M140729X},
  url     = {https://doi.org/10.1137/21M140729X}
}

@article{LiuChristandlVerstraete2007,
  author  = {Liu, Yi-Kai and Christandl, Matthias and Verstraete, Frank},
  title   = {Quantum Computational Complexity of the {$N$}-Representability Problem: {QMA} Complete},
  journal = {Physical Review Letters},
  volume  = {98},
  number  = {11},
  pages   = {110503},
  year    = {2007},
  doi     = {10.1103/PhysRevLett.98.110503},
  url     = {https://doi.org/10.1103/PhysRevLett.98.110503}
}

@article{YuEtAl2021,
  author  = {Yu, Xiao-Dong and Simnacher, Timo and Wyderka, Nikolai and Nguyen, H. Chau and G{\"u}hne, Otfried},
  title   = {A Complete Hierarchy for the Pure State Marginal Problem in Quantum Mechanics},
  journal = {Nature Communications},
  volume  = {12},
  pages   = {1012},
  year    = {2021},
  doi     = {10.1038/s41467-020-20799-5},
  url     = {https://doi.org/10.1038/s41467-020-20799-5}
}

@inproceedings{KobayashiMatsumotoYamakami2003,
  author    = {Kobayashi, Hirotada and Matsumoto, Keiji and Yamakami, Tomoyuki},
  title     = {Quantum {Merlin--Arthur} Proof Systems: Are Multiple Merlins More Helpful to Arthur?},
  booktitle = {Algorithms and Computation},
  series    = {Lecture Notes in Computer Science},
  volume    = {2906},
  pages     = {189--198},
  publisher = {Springer},
  year      = {2003},
  doi       = {10.1007/978-3-540-24587-2_21},
  url       = {https://doi.org/10.1007/978-3-540-24587-2_21}
}

@article{Banach1938,
  author  = {Banach, Stefan},
  title   = {{\"U}ber homogene Polynome in {$(L^2)$}},
  journal = {Studia Mathematica},
  volume  = {7},
  number  = {1},
  pages   = {36--44},
  year    = {1938},
  doi     = {10.4064/sm-7-1-36-44},
  url     = {https://doi.org/10.4064/sm-7-1-36-44}
}

@article{HubenerEtAl2009,
  author  = {H{\"u}bener, Robert and Kleinmann, Matthias and Wei, Tzu-Chieh and Gonz{\'a}lez-Guill{\'e}n, Carlos and G{\"u}hne, Otfried},
  title   = {Geometric Measure of Entanglement for Symmetric States},
  journal = {Physical Review A},
  volume  = {80},
  number  = {3},
  pages   = {032324},
  year    = {2009},
  doi     = {10.1103/PhysRevA.80.032324},
  url     = {https://doi.org/10.1103/PhysRevA.80.032324}
}

@article{WeiGoldbart2003,
  author  = {Wei, Tzu-Chieh and Goldbart, Paul M.},
  title   = {Geometric Measure of Entanglement and Applications to Bipartite and Multipartite Quantum States},
  journal = {Physical Review A},
  volume  = {68},
  number  = {4},
  pages   = {042307},
  year    = {2003},
  doi     = {10.1103/PhysRevA.68.042307},
  url     = {https://doi.org/10.1103/PhysRevA.68.042307}
}

@article{ChristandlEtAl2007,
  author  = {Christandl, Matthias and K{\"o}nig, Robert and Mitchison, Graeme and Renner, Renato},
  title   = {One-and-a-Half Quantum {de Finetti} Theorems},
  journal = {Communications in Mathematical Physics},
  volume  = {273},
  number  = {2},
  pages   = {473--498},
  year    = {2007},
  doi     = {10.1007/s00220-007-0189-3},
  url     = {https://doi.org/10.1007/s00220-007-0189-3}
}

@misc{JWXdeFinetti,
  author        = {Jeronimo, Granha Fernando and Wu, Pei and Xu, Haochen},
  title         = {Optimal Quantum {de Finetti} Theorems via Argmax Rounding},
  year          = {2026},
  eprint        = {2608.02590},
  archivePrefix = {arXiv},
  primaryClass  = {quant-ph},
  doi           = {10.48550/arXiv.2608.02590},
  url           = {https://arxiv.org/abs/2608.02590}
}

@misc{JWXSoS,
  author        = {Jeronimo, Granha Fernando and Wu, Pei and Xu, Haochen},
  title         = {An Argmax Principle for Sum-of-Squares Relaxations on the Sphere},
  year          = {2026},
  eprint        = {2608.02594},
  archivePrefix = {arXiv},
  primaryClass  = {cs.CC},
  note          = {Submitted to SODA},
  doi           = {10.48550/arXiv.2608.02594},
  url           = {https://arxiv.org/abs/2608.02594}
}

@article{BrandaoHarrow2017,
  author  = {Brand{\~a}o, Fernando G. S. L. and Harrow, Aram W.},
  title   = {Quantum {de Finetti} Theorems Under Local Measurements with Applications},
  journal = {Communications in Mathematical Physics},
  volume  = {353},
  number  = {2},
  pages   = {469--506},
  year    = {2017},
  doi     = {10.1007/s00220-017-2880-3},
  url     = {https://doi.org/10.1007/s00220-017-2880-3}
}

@article{AaronsonEtAl2009,
  author  = {Aaronson, Scott and Beigi, Salman and Drucker, Andrew and Fefferman, Bill and Shor, Peter},
  title   = {The Power of Unentanglement},
  journal = {Theory of Computing},
  volume  = {5},
  number  = {1},
  pages   = {1--42},
  year    = {2009},
  doi     = {10.4086/toc.2009.v005a001},
  url     = {https://doi.org/10.4086/toc.2009.v005a001}
}

@article{GharibianSikoraUpadhyay2013,
  author  = {Gharibian, Sevag and Sikora, Jamie and Upadhyay, Sarvagya},
  title   = {{QMA} Variants with Polynomially Many Provers},
  journal = {Quantum Information \& Computation},
  volume  = {13},
  number  = {1--2},
  pages   = {135--157},
  year    = {2013},
  doi     = {10.26421/QIC13.1-2-8},
  url     = {https://doi.org/10.26421/QIC13.1-2-8}
}

@article{BuhrmanEtAl2001,
  author  = {Buhrman, Harry and Cleve, Richard and Watrous, John and de Wolf, Ronald},
  title   = {Quantum Fingerprinting},
  journal = {Physical Review Letters},
  volume  = {87},
  number  = {16},
  pages   = {167902},
  year    = {2001},
  doi     = {10.1103/PhysRevLett.87.167902},
  url     = {https://doi.org/10.1103/PhysRevLett.87.167902}
}

@article{Hoeffding1963,
  author  = {Hoeffding, Wassily},
  title   = {Probability Inequalities for Sums of Bounded Random Variables},
  journal = {Journal of the American Statistical Association},
  volume  = {58},
  number  = {301},
  pages   = {13--30},
  year    = {1963},
  doi     = {10.1080/01621459.1963.10500830},
  url     = {https://doi.org/10.1080/01621459.1963.10500830}
}

@article{Farouki2012,
  author  = {Farouki, Rida T.},
  title   = {The {Bernstein} Polynomial Basis: A Centennial Retrospective},
  journal = {Computer Aided Geometric Design},
  volume  = {29},
  number  = {6},
  pages   = {379--419},
  year    = {2012},
  doi     = {10.1016/j.cagd.2012.03.001},
  url     = {https://doi.org/10.1016/j.cagd.2012.03.001}
}

@article{BochnakSiciak1971,
  author  = {Bochnak, Jacek and Siciak, J{\'o}zef},
  title   = {Polynomials and Multilinear Mappings in Topological Vector Spaces},
  journal = {Studia Mathematica},
  volume  = {39},
  number  = {1},
  pages   = {59--76},
  year    = {1971},
  doi     = {10.4064/sm-39-1-59-76},
  url     = {https://doi.org/10.4064/sm-39-1-59-76}
}

@article{BlierTapp2012,
  author  = {Blier, Hugue and Tapp, Alain},
  title   = {A Quantum Characterization of {NP}},
  journal = {Computational Complexity},
  volume  = {21},
  number  = {3},
  pages   = {499--510},
  year    = {2012},
  doi     = {10.1007/s00037-011-0016-2},
  url     = {https://doi.org/10.1007/s00037-011-0016-2}
}

@article{HarrowMontanaro2013,
  author  = {Harrow, Aram W. and Montanaro, Ashley},
  title   = {Testing Product States, Quantum {Merlin--Arthur} Games and Tensor Optimisation},
  journal = {Journal of the ACM},
  volume  = {60},
  number  = {1},
  pages   = {3:1--3:43},
  year    = {2013},
  doi     = {10.1145/2432622.2432625},
  url     = {https://doi.org/10.1145/2432622.2432625}
}

@misc{ChenDrucker2010,
  author        = {Chen, Jing and Drucker, Andrew},
  title         = {Short Multi-Prover Quantum Proofs for {SAT} without Entangled Measurements},
  year          = {2010},
  eprint        = {1011.0716},
  archivePrefix = {arXiv},
  primaryClass  = {quant-ph},
  doi           = {10.48550/arXiv.1011.0716},
  url           = {https://arxiv.org/abs/1011.0716}
}

\appendix

\section{Analytic Tools}

\subsection{A Bernoulli Concentration Bound}

The following is the Bernoulli specialization of Hoeffding's inequality \cite{Hoeffding1963}.

\begin{lemma}\label{lem:hoeffding}
Let $Y_1,\ldots,Y_t\in\set{0,1}$ be independent random variables having a common mean $p$. For every $a>0$,
\begin{equation}
\Pr\brak*{
\abs*{\frac1t\sum_{h=1}^tY_h-p}\ge a
}
\le
2e^{-2ta^2}.
\notag
\end{equation}
\end{lemma}
\begin{proof}
First, let $Y\in\set{0,1}$ have mean $p$, and define
\begin{equation}
\phi(\lambda)
:=
\ln\EE\brak*{e^{\lambda(Y-p)}}.
\notag
\end{equation}
Under the exponentially tilted distribution, $Y$ remains Bernoulli. Consequently, $\phi''(\lambda)$ is the variance of a Bernoulli random variable and satisfies
\begin{equation}
0\le\phi''(\lambda)\le\frac14.
\notag
\end{equation}
Because $\phi(0)=\phi'(0)=0$, it follows that
\begin{equation}
\EE\brak*{e^{\lambda(Y-p)}}
\le
e^{\lambda^2/8}
\notag
\end{equation}
for every $\lambda\in\R$.

For $\lambda>0$, Markov's inequality and independence give
\begin{align}
\Pr\brak*{\frac1t\sum_{h=1}^tY_h-p\ge a}
&\le
e^{-\lambda ta}
\prod_{h=1}^t
\EE\brak*{e^{\lambda(Y_h-p)}}
\notag \\
&\le
\exp\paren*{-\lambda ta+\frac{t\lambda^2}{8}}.
\notag
\end{align}
Taking $\lambda=4a$ gives $e^{-2ta^2}$. Applying the same argument to the Bernoulli variables $1-Y_h$ gives the lower-tail bound, and adding the two tail probabilities proves the lemma.
\end{proof}

\subsection{Bernstein Polynomials}

We record the derivative identity for the Bernstein basis
\cite{Farouki2012}. For a vector $\ol a=(a_0,\ldots,a_t)$, define its forward difference by
\begin{equation}
(\Delta\ol a)_\ell:=a_{\ell+1}-a_\ell,
\notag
\end{equation}
and let $\Delta^j$ denote the $j$-fold iterate.

For $j\in\set{0,\ldots,t}$, write
\[
(t)_j:=t(t-1)\cdots(t-j+1)=\frac{t!}{(t-j)!},
\qquad
(t)_0:=1,
\]
for the falling factorial.

\begin{lemma}\label{lem:bernstein-derivatives}
Let $t\in\N_{>0}$, let $\ol a=(a_0,\ldots,a_t)\in[0,1]^{t+1}$, and define
\begin{equation}
Q_{\ol a}(p)
:=
\sum_{\ell=0}^t
a_\ell\binom t\ell p^\ell(1-p)^{t-\ell}.
\notag
\end{equation}
For every $j\in\set{0,\ldots,t}$,
\begin{equation}
Q_{\ol a}^{(j)}(p)
=
(t)_j
\sum_{\ell=0}^{t-j}
(\Delta^j\ol a)_\ell
\binom{t-j}{\ell}
p^\ell(1-p)^{t-j-\ell}.
\notag
\end{equation}
Moreover, for every $j\in\set{1,\ldots,t}$ and $p\in[0,1]$,
\begin{equation}
\abs{Q_{\ol a}^{(j)}(p)}
\le
2^{j-1}(t)_j.
\notag
\end{equation}
\end{lemma}
\begin{proof}
Let
\begin{equation}
b_{t,\ell}(p):=\binom t\ell p^\ell(1-p)^{t-\ell},
\notag
\end{equation}
where a basis polynomial having an index outside
$\set{0,\ldots,t}$ is interpreted as zero. Direct differentiation gives
\begin{equation}
\frac{\mathrm d}{\mathrm dp}b_{t,\ell}(p)
=
t\paren[\big]{b_{t-1,\ell-1}(p)-b_{t-1,\ell}(p)}.
\notag
\end{equation}
Multiplying by $a_\ell$, summing, and reindexing yields
\begin{equation}
Q_{\ol a}'(p)
=
t\sum_{\ell=0}^{t-1}
(\Delta\ol a)_\ell b_{t-1,\ell}(p).
\notag
\end{equation}
Iteration proves the derivative identity.

The explicit finite-difference formula is
\begin{equation}
(\Delta^j\ol a)_\ell
=
\sum_{k=0}^j
(-1)^{j-k}\binom jk a_{\ell+k}.
\notag
\end{equation}
For $j\ge1$, the positive and negative binomial coefficients in this alternating sum each have total mass $2^{j-1}$. Since every
$a_{\ell+k}\in[0,1]$,
\begin{equation}
\abs{(\Delta^j\ol a)_\ell}
\le
2^{j-1}.
\notag
\end{equation}
The Bernstein basis polynomials are non-negative and sum to one on $[0,1]$, which proves the derivative bound.
\end{proof}

\subsection{Polarization and Elementary Estimates}

The following lemma is basically the Fourier form of the classical polarization principle that recovers multilinear map from homogenous polynomial \cite{BochnakSiciak1971}.

\begin{lemma}\label{lem:torus-polarization}
Let $\mcV$ be a finite-dimensional complex vector space, let
$k\in\N_{>0}$, and let $B:\mcV^k\to\C$ be a symmetric multilinear form. For
$v_1,\ldots,v_k\in\mcV$ and
$\ol\theta=(\theta_1,\ldots,\theta_k)\in[0,2\pi]^k$, define
\begin{equation}
w_{\ol\theta}
:=
\sum_{r=1}^k e^{i\theta_r}v_r.
\notag
\end{equation}
Then
\begin{equation}
B(v_1,\ldots,v_k)
=
\frac1{k!(2\pi)^k}
\int_{[0,2\pi]^k}
e^{-i(\theta_1+\cdots+\theta_k)}
B(w_{\ol\theta},\ldots,w_{\ol\theta})
\,\mathrm d\ol\theta.
\notag
\end{equation}
In particular, if $B(w,\ldots,w)=0$ for every $w\in\mcV$, then $B=0$.
\end{lemma}
\begin{proof}
By multilinearity,
\begin{align}
&\int_{[0,2\pi]^k}
e^{-i(\theta_1+\cdots+\theta_k)}
B(w_{\ol\theta},\ldots,w_{\ol\theta})
\,\mathrm d\ol\theta
\notag \\
&\quad =
\sum_{\ol r\in[k]^k}
B(v_{r_1},\ldots,v_{r_k})
\int_{[0,2\pi]^k}
\exp\paren*{
-i\sum_{h=1}^k\theta_h
+i\sum_{h=1}^k\theta_{r_h}
}
\,\mathrm d\ol\theta.
\notag
\end{align}
The integral vanishes unless every index in $[k]$ appears exactly once in $\ol r$. The $k!$ surviving tuples are the permutations of
$(1,\ldots,k)$, every surviving integral equals $(2\pi)^k$, and symmetry makes every surviving multilinear value equal to
$B(v_1,\ldots,v_k)$. Dividing by $k!(2\pi)^k$ proves the identity and its consequence.
\end{proof}

\begin{samepage}
\begin{lemma}\label{lem:elementary-estimates}
For integers $1\le j\le t$,
\begin{equation}
j!\ge\paren*{\frac je}^j
\qquad\text{and}\qquad
\binom tj\le\paren*{\frac{et}{j}}^j.
\notag
\end{equation}
For every $\rho\in[0,\frac12]$,
\begin{equation}
\sum_{j=2}^{\infty}\rho^j
=
\frac{\rho^2}{1-\rho}
\le
2\rho^2.
\notag
\end{equation}
\end{lemma}
\end{samepage}
\begin{proof}
We have
\begin{equation}
\ln(j!)
=
\sum_{r=1}^j\ln r
\ge
\int_1^j\ln x\,\mathrm dx
=
j\ln j-j+1
\ge
j\ln j-j.
\notag
\end{equation}
Exponentiating proves the factorial bound, and
\begin{equation}
\binom tj
\le
\frac{t^j}{j!}
\le
\paren*{\frac{et}{j}}^j.
\notag
\end{equation}
The remaining assertion is the sum of a geometric series.
\end{proof}

\section{The Random-Transposition Test}

For $1\le a<b\le N$, let $F_{a,b}$ swap tensor factors $a$ and $b$. Recall that the rejection effect of the uniformly random-pair SWAP test is
\begin{equation}
\rmH_{\mathrm{asym}}
=
\frac1{\binom N2}
\sum_{1\le a<b\le N}
\frac{I-F_{a,b}}2.
\notag
\end{equation}

\begin{theorem}\label{thm:random-transposition-gap}
For every finite-dimensional Hilbert space $\mcH$ and every
$N\in\N_{\ge2}$,
\begin{equation}
\rmH_{\mathrm{asym}}
\succeq
\frac1{N-1}
\paren[\big]{I-\Pi_{\mathrm{sym}}^{(N)}}.
\notag
\end{equation}
In particular,
\begin{equation}
\ker(\rmH_{\mathrm{asym}})
=
\bigvee^N\!(\mcH).
\notag
\end{equation}
\end{theorem}
\begin{proof}
Define the central transposition sum
\begin{equation}
\rmC_N
:=
\sum_{1\le a<b\le N}F_{a,b}.
\notag
\end{equation}
The irreducible complex representations of $\mfS_N$ are indexed by partitions $\lambda\vdash N$. In the Young basis, the Jucys--Murphy element
\begin{equation}
\rmJ_k:=\sum_{a<k}(a,k)
\notag
\end{equation}
acts diagonally with eigenvalue equal to the content $s-r$ of the box $(r,s)$ containing $k$. Since the sum of all transpositions is
$\sum_{k=1}^N\rmJ_k$, it acts on the irreducible representation indexed by $\lambda$ as the scalar
\begin{align}
c_\lambda
&:=
\sum_{(r,s)\in\lambda}(s-r)
\notag \\
&=
\frac12
\sum_{r\ge1}
\lambda_r\paren[\big]{\lambda_r-2r+1}.
\notag
\end{align}

For the one-row partition $(N)$,
\begin{equation}
c_{(N)}=\binom N2.
\notag
\end{equation}
This is the trivial representation, whose isotypic subspace in the tensor-permutation representation is precisely
$\bigvee^N\!(\mcH)$.

Now let $\lambda\ne(N)$. If $\lambda$ has at least two boxes below its first row, move the final box of the lowest non-empty row $r\ge2$ to the end of the first row. This preserves the partition inequalities and increases the content sum by
\begin{equation}
\lambda_1-(\lambda_r-r)
=
\lambda_1-\lambda_r+r
>
0.
\notag
\end{equation}
Repeating this operation leaves the partition $(N-1,1)$. Therefore,
\begin{align}
c_\lambda
&\le
c_{(N-1,1)}
\notag \\
&=
\frac{(N-1)(N-2)}2-1
=
\binom N2-N.
\notag
\end{align}

On the isotypic component indexed by $\lambda$,
\begin{equation}
\rmH_{\mathrm{asym}}
=
\frac12
\paren*{
I-\frac{\rmC_N}{\binom N2}
}.
\notag
\end{equation}
It vanishes on the trivial component. On every non-trivial component, its eigenvalue is at least
\begin{equation}
\frac12
\paren*{
1-\frac{\binom N2-N}{\binom N2}
}
=
\frac1{N-1}.
\notag
\end{equation}
This proves the operator inequality and the assertion about the kernel.
\end{proof}

\section{Finite Sampling}

For probability distributions $\mu,\nu$ on a finite set $\Omega$, write
\begin{equation}
d_{\mathrm{TV}}(\mu,\nu)
:=
\frac12\sum_{\omega\in\Omega}
\abs{\mu(\omega)-\nu(\omega)}.
\notag
\end{equation}

\begin{lemma}\label{lem:finite-sampling}
Let $D:\N\to\N_{>0}$ be polynomially bounded and polynomial-time computable, and let an ideal polynomial-time verifier on inputs of length $n$ make at most $D(n)$ classical random choices. Conditionally on the preceding choices, suppose that every choice $X_h$ satisfies
\begin{equation}
X_h
\sim
\operatorname{Ber}(p_h),
\quad
p_h\in\Q\cap[0,1],
\qquad\text{or}\qquad
X_h
\sim
\operatorname{Unif}(\Omega_h).
\notag
\end{equation}
Suppose that $p_h$ and $|\Omega_h|$ are computable in polynomial time from the input and preceding choices and have polynomial bit length, and that a bijection
$\operatorname{unrank}_h:[|\Omega_h|]\to\Omega_h$ is evaluable in polynomial time. Suppose also that every choice-dependent quantum circuit has a polynomial-size description over $\mcG$ generated in polynomial time. Then, for every polynomial-time computable rational $0<\xi\le1$ of polynomial bit length, there is a polynomial-time uniform circuit whose acceptance probability differs from the ideal value by at most $\xi$ on every witness.

Unordered pairs, $t$-element subsets of $[N]$, and ordered injections $[m]\hookrightarrow[N]$ satisfy the unranking hypothesis for the parameters in \Cref{sec:puresuper-collapse,sec:bellpure-collapse}.
\end{lemma}
\begin{proof}
Fix the input, write $D_n:=D(n)$, and set
\begin{equation}
\kappa
:=
\ceil*{\log_2\paren*{\frac{D_n}{\xi}}}.
\notag
\end{equation}
For a uniform choice from a set of size $K$, let
$k:=\ceil*{\log_2K}+\kappa$, sample $Z$ uniformly from
$\set{0,\ldots,2^k-1}$, and use $1+(Z\bmod K)$ as the rank. Every rank has either $\floor*{2^k/K}$ or
$\ceil*{2^k/K}$ preimages, so
\begin{equation}
\begin{split}
d_{\mathrm{TV}}\paren*{
\operatorname{Law}\paren*{1+(Z\bmod K)},
\operatorname{Unif}([K])
}
&\le
\frac{K}{2^{k+1}}
\le
2^{-\kappa},
\\
\widetilde p
&:=
2^{-\kappa}\floor*{2^\kappa p},
\qquad
d_{\mathrm{TV}}\paren*{
\operatorname{Ber}(p),
\operatorname{Ber}(\widetilde p)
}
\le
2^{-\kappa}.
\end{split}
\notag
\end{equation}
Here, the second line gives the dyadic approximation for a Bernoulli parameter $p$. Let $\mu$ and $\nu$ be the ideal and implemented transcript distributions. A hybrid argument gives
\begin{equation}
\begin{split}
d_{\mathrm{TV}}(\mu,\nu)
&\le
D_n2^{-\kappa}
\le
\xi,
\\
\abs{
\Pr_\mu\brak{\mathrm{accept}}
-
\Pr_\nu\brak{\mathrm{accept}}
}
&=
\abs{\EE_\mu[f]-\EE_\nu[f]}
\le
d_{\mathrm{TV}}(\mu,\nu)
\le
\xi,
\end{split}
\notag
\end{equation}
where $f$ is the conditional acceptance probability of a fixed witness given the transcript. The sampling arithmetic has polynomial-size reversible implementations. Choice-dependent circuit descriptions are computed reversibly, padded, evaluated by an exact program-controlled circuit over the controlled extension of $\mcG$ specified before \Cref{def:qma}, and uncomputed. Retaining the sampled strings through the accepting measurement keeps distinct transcripts orthogonal.

It remains to verify the unranking assertion. If $a$ is the current subset candidate and $s$ elements remain to be chosen, and if a value has just been fixed in position $j$ of an ordered injection, the corresponding block sizes are
\begin{equation}
\binom{N-a}{s-1}
\qquad\text{and}\qquad
(N-j)_{m-j}
\notag
\end{equation}
respectively. Comparing the rank with consecutive blocks constructs a subset in at most $N$ stages and an ordered injection in $m$ stages. Moreover,
\begin{equation}
\log_2\binom Nt
\le
N,
\qquad
\log_2(N)_m
\le
m\log_2N.
\notag
\end{equation}
Thus, all relevant integers have polynomial bit length. Unordered pairs are the case $t=2$, which proves the assertion and completes the proof.
\end{proof}

\newpage

\end{document}